\documentclass[secthm,seceqn,MSNbibl,number,citesort]{arxbj}
\usepackage{uplgreek}
\usepackage{mathbh}

\CRC
\vtexed{Siga}

\aid{0}
\volume{18}
\issue{3}
\pubyear{2012}
\firstpage{883}
\lastpage{913}
\doi{10.3150/11-BEJ359}

\startlocaldefs
\newtheorem{prop}[thm]{Proposition}
\newtheorem{cor}[thm]{Corollary}
\newtheorem{lemma}[thm]{Lemma}
\def\1{\mathbh{1}}
\def\E{\mathbb{E}}
\def\N{\mathbb{N}}
\def\P{\mathbb{P}}
\def\R{\mathbb{R}}
\def\Z{\mathbb{Z}}
\newcommand{\mathds}{\mathbb}
\newcommand{\eqref}[1]{(\ref{#1})}
\newcommand{\fraca}[2]{{#1}/{#2}}
\newcommand{\fracd}[2]{({#1}/{#2})}
\newcommand{\fracc}[2]{{#1}/{(#2)}}
\endlocaldefs

\begin{document}
\begin{frontmatter}

\title{Model selection for weakly dependent time series forecasting}
\runtitle{Model selection for weakly dependent time series forecasting}

\begin{aug}
\author{\inits{P.}\fnms{Pierre} \snm{Alquier}\corref{}\thanksref
{1}\ead[label=e1]{alquier@math.jussieu.fr}} \and
\author{\inits{O.}\fnms{Olivier} \snm{Wintenberger}\thanksref
{2}\ead[label=e2]{wintenberger@ceremade.dauphine.fr}}
\runauthor{P. Alquier and O. Wintenberger}
\address[1]{Laboratoire de Probabilit\'es et Mod\`eles Al\'eatoires,
Universit\'e Paris 7, site Chevaleret,
175, rue du Chevaleret,
75205 Paris Cedex 13,
France,
and CREST, Laboratoire de Statistique,
3, avenue Pierre Larousse,
92240 Malakoff,
France. \printead{e1}}
\address[2]{CEREMADE, Universit\'e Paris Dauphine,
Place du Mar\'echal De Lattre De Tassigny,
75775 Paris Cedex 16,
France. \printead{e2}}
\end{aug}

\received{\smonth{2} \syear{2009}}
\revised{\smonth{1} \syear{2011}}

%
\begin{abstract}
Observing a stationary time series, we propose a two-steps procedure for
the prediction of its next value.
The first step follows machine learning theory paradigm and consists in
determining a set of possible predictors as
randomized estimators in (possibly numerous) different predictive
models. The second
step follows the model selection paradigm and consists in choosing one
predictor with
good properties among all the predictors of the first step. We study
our procedure for
two different types of observations: causal Bernoulli shifts and
bounded weakly dependent
processes. In both cases, we give oracle inequalities: the risk of the
chosen predictor is
close to the best prediction risk in all predictive models that we
consider. We apply our
procedure for predictive models as linear predictors, neural networks
predictors and
nonparametric autoregressive
predictors.
\end{abstract}

%

\end{frontmatter}

\section{Introduction}
When observing a time series, one crucial issue is to predict the
(nonobserved) first future value using the
observed past values. Since the seventies, different model selection
procedures have been
studied for inferring how many observed past values are needed for
predicting the next value. Procedures as AIC \cite{Akaike1973}, BIC
(Schwarz \cite{bic}) and APE (Ing \cite{Ing2007})
are used by practitioners to select a reasonable linear predictor. When
the observations satisfy a linear model, those procedures are proved to
be asymptotically efficient (see Ing \cite{Ing2007} for more details).

In the same time, the progress of statistical learning theory in the
i.i.d. setting brought new perspectives in
model selection (see Vapnik \cite{Vapnik} and Massart \cite{Massart}
among others). Machine-learning procedures allow to
choose a predictor among a family, with the guarantee that this
predictor performs almost
as well as the best possible predictor of the family (called the
oracle). Such results are called oracle inequalities; they provide
guarantees on the
quality of the prediction without any parametric assumption on the observations.

Few works have been done in the context of dependent observations. The
machine learning theory was used successfully in the time series
prediction context by Modha and Masry
\cite{modha} and Meir \cite{meir}. However, their procedure relies on
the knowledge of the $\alpha$-mixing coefficients. To our knowledge,
there is no efficient estimation of this coefficients and their
procedure seems difficult to use in practice. Baraud \textit{et al.} \cite
{Baraud} use the
model selection point of view to perform regression and auto-regression
on dependent observations.
They prove powerful oracle inequalities when the observations satisfy
an additive auto-regressive
model.
When the observations are Harris
recurrent Markov chains, Lacour \cite{LACOUR} gives also oracle
inequalities for a procedure completely
free of the dependence properties. An alternative point of view is
provided by the theory of individual sequences prediction
(see Lugosi and Cesa-Bianchi \cite{Lugosi} or Stoltz \cite{Stoltz}).
In these works, no assumption on the observations -- not even a
stochastic assumption -- is done and oracle inequalities are given.

In this paper, our objectives are the following:
\begin{enumerate}[(3)]
\item[(1)] to build various predictors of different forms and using
different numbers
of past observations,
\item[(2)] to select one of these predictors \textit{without any assumption
on the distribution of the observations},
\item[(3)] to prove oracle inequalities under weak assumptions on the
observed time series.
\end{enumerate}
In the end of this Introduction, let us fix the mathematical framework (see
also Meir \cite{meir} for more details).

Let us observe $(X_1,\ldots,X_n)$ from a stationary time series
$X=(X_{t})_{t\in\mathds{Z}}$ distributed as $\pi_0$ on $\mathcal
{X}^\Z$ where $\mathcal{X}$ is
an Hilbert space equipped with its usual norm $\|\cdot\|$. Fix a
(possibly large) family of predictors
$ \{f_{\theta},\theta\in\Theta\} $: for any $\theta$ and any
$t$, $f_{\theta}$ applied to the past values $(X_{t-1},X_{t-2},\ldots
, X_1)$ is a possible prediction
of $X_{t}$. We discretize the family of predictors by the number $p$ of
past values they use.
Thus, we assume that
\[
\Theta= \bigcup_{p=1}^{ \lfloor\fraca{n}{2}  \rfloor
}\Theta_{p},
\]
where the $\Theta_{p}$ are disjoint in order that for any $\theta\in
\Theta$, there is only one $p$ such that
$\theta\in\Theta_{p}$. Now, for any
$\theta\in\Theta_{p}$, $f_{\theta}$ is a function $\mathcal
{X}^{p}\rightarrow\mathcal{X}$ and at
any time $t$, $f_{\theta}(X_{t-1},\ldots,X_{t-p})$ is a prediction of
$X_{t}$ according
to $\theta$ and denoted $\hat{X}_{t}^{\theta}$. As the predictor
$f_\theta$ may take different
forms (linear functions, neural networks$,\ldots$), we write
\[
\Theta_{p} = \bigcup_{\ell=1}^{m_{p}}\Theta_{p,\ell}
\]
for a given $m_{p}\in\{1,\ldots,n\}$. Finally, the risk of the
prediction, $R(\theta)$,
is defined by
\[
R(\theta)=\pi_0 [ \|f_{\theta}(X_{t-1},\ldots,X_{t-p})-X_{t}
 \| ]
= \pi_0 [ \|\hat{X}_{t}^{\theta}-X_{t}  \| ],
\]
where here and all along the paper $\pi[h]=\int h\,\mathrm{d}\pi$ for any
measure $\pi$ and any integrable function~$h$. Note that $R(\theta)$
does not depend on $t$ as $X$ is stationary.

The mathematical counterparts of the points (1), (2) and (3) of our
objectives are the following. The point (1) corresponds to build, on
the basis of the observations, an estimator $\hat{\theta} _{p,\ell}$
in each model $\Theta_{p,\ell}$,
for $1\leq p \leq\lfloor n/2 \rfloor$ and $1\leq\ell\leq m_{p}$.
The point (2) consists in defining a procedure to choose a $\hat
{\theta}$ among all the possible $\hat{\theta} _{p,\ell}$.
Finally, point (3) is achieved by proving that $R(\hat{\theta})$ is
close to $\inf_{\theta\in\Theta}R(\theta)$. To attain these
objectives, we use the PAC-Bayesian paradigm (introduced by
Shawe-Taylor and Williamson \cite{STW97} and McAllester \cite{McAllester}).
Using this approach, Catoni \cite{Classif,Cat7,manuscrit}, Audibert
\cite{AudibertReg}, Alquier \cite{Alquier2008},
Tsybakov and Dalalyan \cite{arnak} solve points (1),
(2) and (3) simultaneously for various regression and classification
problems in the i.i.d. setting. In this paper, we build a procedure
that gives a predictor $\hat\theta$ satisfying, under general
conditions on $X$ and
with probability at least $1-\varepsilon$,
\[
R(\hat{\theta})\le
\inf_{ d_{p,\ell}\leq n}  \Biggl\{
\min_{\theta\in\Theta_{p,\ell}}R (\theta ) +
{\rm cst} \cdot \sqrt{\frac{d_{p,\ell}}{n}} \log^{5/2}(n)  \Biggr\} +
{\rm cst} \cdot
\frac{\log\fracd{1}{\varepsilon}}{\sqrt{n}},
\]
where ${\rm cst} >0$ is an explicit constant and $d_{p,\ell}$ an
estimate of the complexity of $\Theta_{p,\ell}$.

To obtain such oracle inequalities, we use
sharp estimates (close to the ones in the i.i.d. case) on the Laplace
transforms of the partial sums in dependent settings. For bounded
observations, we use the $\theta_\infty$-coefficients
(see \cite{Dedecker2007a}), introduced in Rio \cite{Rio2000a} as the
$\gamma$-mixing coefficients. These coefficients generalize the
uniform mixing ones. For unbounded observations, we
use the causal Bernoulli shifts representation. It includes all
classical linear ARMA models and also the more general chains with
infinite memory introduced by Doukhan
and Wintenberger \cite{Doukhan2008}. These bounded and unbounded
dependent frameworks are not comparable with the $\beta$ or $\alpha
$-mixing ones as they include
some dynamical systems that are not mixing, see Andrews \cite
{Andrews1984} and Dedecker and Prieur
\cite{Dedecker2005} for details.
Finally, it is important to note that our prediction procedure is the
same for the two dependence frameworks and
and does not depend on any unknown
dependence coefficient. It is an advantage of our approach because it
is impossible to
estimate efficiently the dependence coefficients we use.

The paper is organized as follows: First, the prediction procedure is
detailed in Section
\ref{sec::pl}; Second, the assumptions on the observed time series and
the corresponding oracle inequalities are given in Section
\ref{sec::res}. In Section \ref{sec::app1}, are given some examples
of time series for which these oracle inequalities hold. Our procedure
applied on some possible prediction models are given in Section \ref
{sec::app2}. Linear predictors (with simulations), neural networks
predictors and
non-parametric predictors are considered. Finally, the complete proofs
are collected in Section \ref{sec::pf}.

\section{The prediction procedure}
\label{sec::pl}

We observe $(X_1,\ldots,X_n)$ from a stationary time series
$X=(X_{t})_{t\in\mathds{Z}}$
distributed as $\pi_0$ on $\mathcal{X}^\Z$ where $\mathcal{X}$ is
an Hilbert space equipped with its
usual norm $\|\cdot\|$. We fix a family of predictors $\{f_{\theta
},\theta\in\Theta\}$ with
\[
\Theta= \bigcup_{p=1}^{ \lfloor\fraca{n}{2}  \rfloor
}\Theta_{p}
= \bigcup_{p=1}^{ \lfloor\fraca{n}{2}  \rfloor}
 \Biggl(\bigcup_{\ell=1}^{m_{p}}\Theta_{p,\ell} \Biggr)
\]
such that $m_p\ge n$ and $p(\theta)$ is the only $p$ such that $\theta
\in\Theta_{p}$. For any $\theta\in\Theta$, we denote
$\hat{X}_{t}^{\theta}=f_{\theta}(X_{t-1},\ldots,X_{t-p})$ and $R(\theta
) = \pi_0 [ \|
\hat{X}_{t}^{\theta}-X_{t}  \| ]$.

\subsection{The Lipschitz predictors}

Let $M$ denotes the set of all possible pairs $(p,\ell)$:
\[
M = \bigcup_{p=1}^{ \lfloor\fraca{n}{2} \rfloor} \{p\}
\times\{1,\ldots,m_{p}\}.
\]
Let $\mathcal{T}$ be a $\sigma$-algebra on $\Theta$ and $\mathcal
{T}_{p,\ell}$ be its restriction
to $\Theta_{p,\ell}$ for any $(p,\ell)\in M$. For any $(p,\ell)\in
M$, we assume that $\Theta_{p,\ell}$ is a compact subset of $\R^q$ for
some $q<\infty$ ($q$ depends on $(p,\ell)$) and that there
exists $(a_{j}(\theta))_{j\in\{1,\ldots,p\}}$ satisfying, for any
$(x_{1},\ldots,x_{p}),(y_{1},\ldots,y_{p})\in\mathcal{X}^{p}$, the relation
%
\begin{equation}
\label{modellip}
 \|f_{\theta}(x_{1},\ldots,x_{p}) - f_{\theta
}(y_{1},\ldots,y_{p}) \|
\leq
\sum_{j=1}^{p} a_{j}(\theta)  \|x_{j}-y_{j} \|.
\end{equation}
In order to bound the volatility of the predictors uniformly on $M$, we
assume that
%
\begin{equation}
\label{modellip2}
L:=\sup_{(p,\ell)\in M}
\sup_{\theta\in\Theta_{p,\ell}}\sum_{j=1}^{p} a_{j}(\theta)
  \qquad \mbox{satisfies} \qquad   L\leq\log(n) -1.
\end{equation}

\subsection{The complexity of $\Theta_{p,\ell}$}

To control the complexity of each $\Theta_{p,\ell}$
we
assume that, for all $(p,\ell)\in M$, there exist a probability
measure $\pi_{p,\ell}$ on the measurable space
$ (\Theta_{p,\ell},\mathcal{T}_{p,\ell} )$ and a constant
$1\le d_{p,\ell}<\infty$ satisfying
%
\begin{equation}
\label{dimension}
\sup_{\gamma>e}  \biggl\{
\frac{-\log\int_{\Theta_{p,\ell}}
 [\exp (-\gamma (R(\theta)-R (\overline{\theta
}_{p,\ell} ) ) ) ]\,\mathrm{d}\pi_{p,\ell}(\theta) }
{\log(\gamma)}  \biggr\} \leq d_{p,\ell}.
\end{equation}
Here $\overline{\theta}_{p,\ell}=\arg\min_{\Theta_{p,\ell}}R$
for any $(p,\ell)\in M$. The parameter $d_{p,\ell}$ is linked with classical
complexities as the Vapnik dimension and entropy measures. In this
paper, we only investigate the case where $\pi_{p,\ell}$ is the
Lebesgue measure on $\Theta_{p,\ell}$. We have the following result.
\begin{prop}\label{propdim}
Let $q\in\mathds{N}^{*}$, $x>0$ and $\mathcal B_{x}^{q}$ be
the closed $\ell^1$-ball in $\mathds{R}^{q}$ of radius $x>0$ and
centered at $0$.
If $\Theta_{p,\ell}=\mathcal B_{c_{p,\ell}}^{q}$ for $c_{p,\ell}>0$
and $\theta\to R(\theta)$ is a $C$-Lipschitz function then we have:
%
\begin{equation}\label{eqdim}
d_{p,\ell}\le q\times \biggl(1+\log \biggl(c_{p,\ell} \biggl(\frac
{Ce}{q}\vee
\frac{1}{c_{p,\ell}-\|\overline\theta_{p,\ell}\|} \biggr)
 \biggr) \biggr).
\end{equation}
\end{prop}

The proof of this result is given at the end of Section \ref{prooflems}.
Predictive models where complexity $d_{p,\ell}$ is estimated are given in
Section \ref{sec::app2}.

\subsection{The empirical risk}

As the risk $R(\theta)$ cannot be computed, we use its empirical
counterpart $r_{n}(\theta)$:
\[
r_n(\theta) =
\frac{1}{n-p (\theta )} \sum_{t=p (\theta )+1}^{n}
 \|X_{t}-\widehat{X}^{\theta}_{t} \|.
\]

\subsection{The randomized estimators}

For any $(p,\ell)\in M$,
our randomized estimators $\tilde{\theta}_{p,\ell}^{\lambda}$ is
drawn randomly through a Gibbs measure
\[
\tilde{\theta}_{p,\ell}^{\lambda} \sim\pi_{p,\ell}\{-\lambda
r_n\}.
\]
We recall that for any measure $\pi$ and any measurable function $h$
such that $\pi[\exp(h)]<+\infty$, the Gibbs measure denoted $\pi\{
h\}$ is defined by the relation:
%
\begin{equation}\label{Gibbsmeasure}
\frac{\mathrm{d}\pi\{h\}}{\mathrm{d}\pi}(\theta) = \frac{\exp(h(\theta))}{\pi
[\exp( h)]}.
\end{equation}
Here the parameter $\lambda$ is called the temperature (this
terminology comes from the statistical thermodynamics). For $n\ge8e
(1+L)$, $\lambda$ takes values in a finite grid $\mathcal{G}_{p,\ell
}$ defined as
\[
\mathcal{G}_{p,\ell} =  \biggl\{ g_{1}\frac{\sqrt{d_{p,\ell}n}\log
(d_{p,\ell}n)}{(1+L)\log^{3/2}(n)},\ldots,
g_{n_{0}}\frac{\sqrt{d_{p,\ell}n}\log(d_{p,\ell}n)}{(1+L)\log
^{3/2}(n)} \biggr\}
\cap \biggl[2e,\frac{n}{4(1+L)} \biggr] ,
\]
where $\check{c}\leq g_{1}<\cdots <g_{n_{0}}\leq\hat{c}$ with $2\leq
n_{0} \leq n$ and $0<\check{c}<2/(1+L)<2e(1+L)<\hat{c}<\infty$.
Remark that when $\lambda$ grows,
$\pi_{p,\ell}\{-\lambda r_n\}$ tends to concentrate around the
minimizer of the empirical risk.

\subsection{The model selection}
One way to select a predictor is to choose the minimizer of the
penalized empirical risk $\arg\min_{p,\ell} [r_n(\tilde{\theta
}_{p,\ell}^{\lambda})
+\operatorname{pen}(p,\ell,\lambda)]$, for some well chosen penalization $\operatorname{pen}(p,\ell,\lambda)$, see Massart \cite{Massart}. Here we consider
$\hat{\theta} = \tilde{\theta}_{\hat{p},\hat{\ell}}^{\hat
{\lambda}}$ where
\[
(\hat{p},\hat{\ell},\hat{\lambda}) = \arg\mathop{\mathop{\min}_{(p,\ell)\in M
}}_{
\lambda\in\mathcal{G}_{p,\ell}}
\hat{R} (p,\ell,\lambda ).
\]
The model criterion $ \hat{R} (p,\ell,\lambda )$ is given
by the PAC-Bayesian approach:
\[
\hat{R} (p,\ell,\lambda ) =
- \frac{1}{\lambda}\log\int_{\Theta_{p,\ell}} \exp
(-\lambda r_n(\theta) )\,\mathrm{d}\pi_{p,\ell}(\theta)
+\frac{1}{\lambda}\log \biggl(n  \biggl\lfloor\frac{n}{2}
\biggr\rfloor
m_{p}  \biggr)
+ \frac{\lambda(1+L)^2 \log^{3}(n)}{n (1-{p}/{n})^{2}}.
\]

\section{Main results}\label{sec::res}

In order to prove that $R(\hat{\theta})$ is close to $\inf_{\theta
\in\Theta}R(\theta)$ with high probability,
we restrict our study to two different contexts. Note that $\hat
{\theta}$ is defined independently of these contexts and that a
practitioner may compute
our predictor on any observed time series.

\subsection{Bounded weakly dependent processes {(WDP)}}
In this case, $X$ is bounded, that is, $\|X\|_\infty:=\sup_{t}\|
X_{t}\| <\infty$. We use the $\theta_{\infty,n}(1)$-coefficients in
Dedecker \textit{et al.}
\cite{Dedecker2007a}, a version of the $\gamma$-mixing of Rio \cite
{Rio2000}) adapted to stationary
time series. If $Z$ is a bounded variable in $\mathcal X^q$ ($q\ge1$)
defined on $(\Omega,\mathcal{A},\mathds{P})$,
for any $\sigma$-algebra $\mathfrak{S}$ of $\mathcal{A}$ we have:
\[
\theta{_\infty}(\mathfrak{S},Z)
=\sup_{f\in\Lambda_{1}} \| |\mathds{\mathds{E}}
(f(Z) |\mathfrak{S}  )
-\mathds{\mathds{E}} (f(Z) ) | \|_{\infty},
\]
where $\Lambda_{1}$ is the set of real $1$-Lipschitz functions on
$\mathcal X^q$ equipped with the norm
$\|z\|=\sum_{i=1}^q \|z_i\|$. Let us define the $\sigma$-algebra
$\mathfrak{S}_{p}
=\sigma(X_{t},t\leq p)$ for any $p\in\Z$ and the coefficients
\[
\theta_{\infty,k}(1)=\sup \{{\theta_\infty} (\mathfrak
{S}_p,(X_{j_1},\ldots,X_{j_\ell}) ),  p+1\le
j_1<\cdots<j_\ell, 1\leq\ell\leq k \}.
\]
Moreover, assume that there is a constant $\mathcal{C}>0$ such that
for any $n$,
$\theta_{\infty,n}(1)<\mathcal{C}$ (the short memory condition).
Causal Bernoulli shifts with bounded innovations, uniform $\varphi
$-mixing sequences and dynamical systems
are classical $\theta_\infty$ weakly-dependent examples, see Section
\ref{sec::app1} for more details. In this context, we prove the
following oracle inequality.
\begin{thm}\label{mainthm}
Under {({WDP})} and condition \eqref{dimension},
there are explicit constants
\[
({\rm cst}_{1},{\rm cst}_{2}) = \operatorname{cst}(\check{c},\hat
{c},L,\mathcal{C},\|X_{0}\|_{\infty})
\]
such that for all $n\ge8e (1+L)$ with probability at least
$1-\varepsilon$
\begin{eqnarray*}
R(\hat{\theta})&\le&
\inf_{  d_{p,\ell}\leq n}  \Biggl\{
\min_{\theta\in\Theta_{p,\ell}}R (\theta ) +
{\rm cst}_{1} \cdot \sqrt{\frac{d_{p,\ell}}{n}} \log^{5/2}(n)  \Biggr\}
+ {\rm cst}_{2}\cdot
\frac{\log\fracd{1}{\varepsilon}}{\sqrt{n}}
\\
&&{}+ 4(1+L)  \biggl(\frac{(\|X_{0}\|_{\infty}+\mathcal C)^{2}}{2} - \log
^{3}(n) \biggr)_{+}.
\end{eqnarray*}
\end{thm}

The proof of this result is given in Section \ref{proofmain} page
\pageref{proofmain}.

\subsection{Causal Bernoulli shifts {(CBS)}}
Let $\mathcal{X}'$ be some Banach space equipped with a norm also
denoted $\|\cdot\|$.
Let $H\dvtx {\mathcal X'}^{\N}\mapsto\mathcal X$ be a
satisfying, for some sequence
$(a_j(H))_{j\in\mathds N}$, and for any $v=(v_j)_{j\in\mathds{N}}$,
$v'=(v'_j)_{j\in\mathds{N}}\in
{\mathcal X'}^{\mathds N}$, the relations:
%
\begin{equation}
\label{condlip1}
 \|H(v)-H(v') \|
\le\sum_{j=0}^\infty a_j(H) \|v_j-v'_j\|,
\end{equation}
 with
\begin{equation}
\label{condsum}
 \sum_{j=0}^\infty ja_j(H)<+\infty.
\end{equation}
We denote $\sum_{j=0}^{\infty} a_j(H)
:=a(H)$, $\sum_{j=0}^\infty ja_j(H)=\tilde a(H)$. The causal Bernoulli
shifts are defined by the relation
\[
X_t=H(\xi_{t},\xi_{t-1},\xi_{t-2},\ldots)  \qquad \forall t\in\Z,
\]
where $\xi_{t}$ for ${t\in\Z}$ are i.i.d. variables called the
innovations and distributed as $\mu$. We assume that we can choose, by quantile
transformation, innovations that admit a finite Laplace transform
$\mu[\exp(c^\ast\|\xi_0\|)]:=\Psi(c^\ast)<+\infty$
(the Cramer condition) for $c^\ast\geq a(H)$.
Classical examples
of such processes are causal linear ARMA models and chains with
infinite memory with low-tail innovations,
see Section \ref{sec::app1} for more details. In this context, we
prove the following oracle inequality

\begin{thm}\label{mainthm2}
Under {(CBS)} and condition \eqref{dimension},
there are explicit constants
\[
({\rm cst}'_{1},{\rm cst}'_{2}) = \operatorname{cst}'(\check{c},\hat
{c},L,a(H),\tilde{a}(H),\Psi(1))
\]
such that for all $n\ge8e (1+L)$ with probability at least
$1-\varepsilon$
\begin{eqnarray*}
R(\hat{\theta})&\le&
\inf_{  d_{p,\ell}\leq n} \Biggl\{ \min_{\theta\in\Theta
_{p,\ell}}R (\theta )
+
{\rm cst}_{1}' \cdot \sqrt{\frac{d_{p,\ell}}{n}} \log^{5/2}(n)
\Biggr\} + {\rm cst}_{2}'\cdot
\frac{\log\fracd{1}{\varepsilon}}{\sqrt{n}}
\\
&&{}+ \sqrt{\frac{d_{\hat{p},\hat{\ell}}}{n}}\log(d_{\hat{p},\hat
{\ell}}n)4(1+L)\\
&&{}\times\hat c
 \biggl(4a(H)\Psi(a(H)) + 2\log^{2}(n) \biggl(1+\frac{\tilde
{a}(H)}{a(H)} \biggr)^{2} - \log^{3}(n)  \biggr)_{+}.
\end{eqnarray*}
\end{thm}

The proof of this result is given in Section \ref{proofmain2}
page \pageref{proofmain2}.

\subsection{Comments on the results}\label{comments}
The constants are roughly (but explicitly) estimated in the proofs, see
Sections \ref{proofmain} and \ref{proofmain2}. For example, we obtain
\[
{\rm cst}_1\le(1+L) \biggl( \frac{6}{\check{c}}
+ 8 \hat{c}  (1+\|X_{0}\|_{\infty} + \mathcal C )^{2}
 \biggr)  \quad \mbox{and} \quad  {\rm cst}_2\le\frac{7(1+L)}{\check c}.
\]
For $n$ sufficiently large, the last terms in the oracle inequalities
vanish. Then it exists a constant $C>0$ such that under {(WDP)} or
{(CBS)} for all $n\ge8e (1+L)$ with probability at least
$1-\varepsilon$:
\[
R(\hat{\theta})\le
\inf_{  d_{p,\ell}\leq n}  \Biggl\{
\min_{\theta\in\Theta_{p,\ell}}R (\theta ) +
C\sqrt{\frac{d_{p,\ell}}{n}} \log^{5/2}(n)  \Biggr\} + C
\frac{\log\fracd{1}{\varepsilon}}{\sqrt{n}}.
\]
Similar oracles inequalities have already been proved by Modha and
Masry \cite{modha} and Baraud \textit{et al.}
\cite{Baraud}. These inequalities are given in expectation while ours
are true with high probability. Remark that integrating our oracle inequalities
with respect to $\varepsilon$ leads to a result in expectation: there
exists a constant
$C>0$ independent of $n$ such that in both {(WDP)} and {(CBS)} cases
\[
\pi_{0}  [R(\hat{\theta}) ] \le
\inf_{  d_{p,\ell}\leq n} \Biggl \{
\min_{\theta\in\Theta_{p,\ell}}R (\theta ) +
C\sqrt{\frac{d_{p,\ell}}{n}} \log^{5/2}(n)  \Biggr\}.
\]
The converse is not true: results in expectation
do not lead to results that hold with high probability.

It is difficult to compare our oracle inequalities with the ones in
\cite{modha} and \cite{Baraud}. Unlike our paper,
those articles
deal with the quadratic risk and ($\beta-$ or $\alpha-$) mixing time
series. However, remark that
the additional terms in our oracle inequalities are proportional to
$\sqrt{d_{p,\ell}/n}$, the rate in the i.i.d. case,
times a term $\log^{5/2}(n)$ term. Baraud \textit{et al.}
\cite{Baraud} obtain an oracle inequality for the quadratic risk with
the same rate than in the i.i.d. case, while
the one in Modha and Masry \cite{modha}
suffers a loss $(n/d_{p,\ell})^{c}$ for some $c>0$.

\section{Examples of time series satisfying {(WDP)} or {(CBS)}}
\label{sec::app1}

We present several examples of time series satisfying {(WDP)} or
{(CBS)}.

\subsection{Causal Bernoulli shifts}

Causal Bernoulli shifts are stationary time series that admit the representation
%
\begin{equation}\label{cbs}
X_t=H(\xi_{t},\xi_{t-1},\xi_{t-2},\ldots) \qquad  \forall t\in\Z,
\end{equation}
where the $\xi_t$ are i.i.d. variables called innovations. Almost all
known stationary and ergodic processes
have this form. However, we work here under the restrictive assumption
\eqref{condlip}. Remark that under
this Lipschitz condition the existence of the stationary time series
$(X_t) $ follows from \eqref{cbs}
and it satisfies the Cramer condition as soon as the innovations do.
Some examples of causal Bernoulli
shifts are presented below.

\subsubsection{Linear models}

Let $(X_t)$ be a real time series admitting the $\operatorname{MA}(\infty$) representation
\[
X_t=\sum_{j=0}^\infty a_j\xi_{t-j} \qquad  \mbox{with } \sum
_{j=0}^\infty j |a_j| <+\infty.
\]
Then it satisfies {(CBS)} if the i.i.d. innovations $\xi_t$
satisfy the Cramer condition. As an example, there is any causal
$\operatorname{AR}(\infty$) model $X_t=\phi_0+\sum_{j=1}^\infty
\phi_jX_{t-j}+\xi_t$ with $\phi(z)=1-\sum_{j=1}^\infty\phi_j z^j$
that have no root for $|z|\le1$ (such that causal $\operatorname{ARMA}(p,q)$
models). Indeed, as $\phi$ is
a real analytic function on the unit disc, $1/\phi$ is a well a real
analytic function
$1/\phi(z)=\sum_{j=1}^\infty\psi_j z^j$ with the coefficients $
\psi_j$ that decrease
exponentially fast (i.e., \eqref{condsum} is automatically satisfied).

\subsubsection{Chains with infinite memory}

Chains with infinite memory is a class of time series $ (X_{t}) $
introduced by Doukhan and
Wintenberger \cite{Doukhan2008} as the solution of the equation
%
\begin{equation}\label{eq::rec}
X_t=F(X_{t-1},X_{t-2},\ldots;\xi_t) \qquad \mbox{almost surely}
\end{equation}
for some function $F\dvtx \mathcal{X}^{(\mathds{N}\setminus\{0\})}\times
\mathcal{X}'\to\mathcal{X}$. Assume also that
for some $u>0$, for all $x=(x_k)_{k\in\mathds{N}\setminus\{0\}}$,
$x'=(x'_k)_{k\in\mathds{N}
\setminus\{0\}}\in\mathcal{X}^{\mathds{N}\setminus\{0\}}$ with
$x_k=x'_k=0$ for all
$k>N$ for some $N>0$, the following condition holds
%
\begin{equation}
\label{condlip}
 \|F(x;y)-F(x';y') \|
\le\sum_{j=1}^\infty a_j(F) \|x_j-x'_j\|+u\|y-y'\|,
\end{equation}
with
\begin{equation}
\label{condcontract}
\sum_{j=1}^{\infty}a_j(F)
:=a(F)<1.
\end{equation}
Many non linear econometrics time series are chains with infinite
memory. The following proposition gives sufficient assumptions such
that chains with infinite memory satisfy {(CBS)}.
\begin{prop}\label{unboundedprop}\label{lemma_ex_1}
Under \eqref{condlip} and \eqref{condcontract} there exists a unique
solution $(X_t)$ of equation \eqref{eq::rec}
satisfying {(CBS)} if $\xi_0$ satisfies the Cramer condition.
\end{prop}

The proof of Proposition \ref{lemma_ex_1} is given in Section \ref
{proofobservations}.

\subsection{Weakly dependent processes}
\label{secwdp}

\subsubsection{Bounded causal Bernoulli shifts}

Bounded causal Bernoulli shifts are examples of time series satisfying
{(WDP)}.
\begin{prop}\label{taucbs}
Under condition \eqref{condlip} and \eqref{condsum}, any solution of
the equation \eqref{cbs} is bounded by
$2a(H)\|\xi_0\|_\infty$ and is weakly dependent {({WDP})} with
$\mathcal C= 2 \|\xi_0\|_\infty\tilde a(H). $
\end{prop}

The proof of this already known result is given in Section \ref
{proofobservations} for completeness. Below are
presented two examples of time series satisfying {(WDP)} that are
not bounded causal Bernoulli shifts.

\subsubsection{Uniform $\varphi$-mixing processes}

Let us recall the definition of the $\varphi$-mixing coefficients
introduced in Ibragimov
\cite{Ibragimov1962};
\[
\varphi(r)=\sup_{(A,B)\in \mathfrak{S}_0\times\mathfrak
{F}_r}|\pi(B/A)-\pi(B)|,
\]
where $\mathfrak{F}_r=\sigma(Y_t,t\ge r)$. The class of $\varphi
$-mixing processes gives examples of time series
that satisfied {(WDP)}.
\begin{prop}\label{lemma_ex_2}
If $(X_t)$ is a stationary bounded process, then it satisfies {({WDP})} with
\[
\theta_{\infty,n}(1)\leq2\|X_0\|_\infty\sum_{r=1}^n \varphi(r).
\]
\end{prop}

The proof of this already known result is given in Section \ref
{proofobservations} for completeness. Remark that $(X_t)$ satisfies the
short memory condition as soon as $(\varphi(r))$ is summable. All
uniform ergodic Markov chains are examples of $\varphi$-mixing
processes with short memory, see Doukhan \cite{Doukhan1994}.

\subsubsection{Dynamical systems on $[0,1]$}
The $\operatorname{AR}(1)$ process $X_t=2^{-1}(X_{t-1}+\xi_t)$ with $\xi_t$
Bernoulli distributed is not mixing, see
\cite{Andrews1984} for more details. Through a reversion of the time,
it can be viewed as a dynamical
system $X_{t}=T(X_{t+1})$ where $T(x)=2x$ if $0\le x<1/2$, $T(x)=2x-1$
if $1/2\le x\le1$. Dedecker and
Prieur \cite{Dedecker2005} extended this counter-example to processes
$(X_t)$ such that
$X_t=T(X_{t+1})$ where $T$ is an expanding map on $[0,1]$, see Section
4.4 of \cite{Dedecker2005} for a
proper definition. Then $(X_t)$ satisfies
{(WDP)} with $
\mathcal C= K\sigma/(1-\sigma)$ where $K>0$, $0\le\sigma<1$, see
Section 7.2 of \cite{Dedecker2005}.

\section{Examples of predictors}
\label{sec::app2}

We give some examples of Lipschitz predictors where we can estimate the
complexity of the $\Theta_{p,\ell}$ and then apply our main results.
In this section, $C>0$ is a constant independent of $\varepsilon$ and
$n$ that may be different from one inequality to another.

\subsection{Linear predictors}

\label{arex}
Let $\mathcal X =\R$ and we consider predictors of the form:
\[
f_{\theta}(X_{t-1},\ldots,X_{t-p}) =\theta_0+ \sum_{i=1}^{p}\theta
_{i}X_{t-i},
\]
where $\theta\in\Theta_{p}\subset\R^{p+1}$ with
\[
\Theta_{p} = \Theta_{p,1}
=  \Biggl\{\theta\in\mathds{R}^{p+1},  \|\theta\|_{1}=\sum
_{i=0}^{p}|\theta_{i}|\leq B  \Biggr\}
\]
for some $B>0$ ($m_p=1$ for all $p$ and we omit the index $\ell$). Using
Proposition \ref{propdim} it follows that
\[
d_{p}\le(p+1)\log \biggl(eB \biggl(\frac{e}{p+1}\vee\frac{1}{B-\|
\overline\theta_{p}\|} \biggr)
 \biggr),
\]
where $\overline\theta_{p}=\arg\min_{\Theta_{p}}R(\theta)$. As a
consequence of Theorems \ref{mainthm} and \ref{mainthm2}, we obtain
the following corollary.
\begin{cor}\label{cor1}
If $ \|\overline{\theta}_{p}\|_1 \leq B-e/(p+1) $ for all $p\ge0$,
then, under {({WDP})} or {(CBS)}, for all $n\ge8e (1+L)$
with probability at least $1-\varepsilon$:
\[
R(\hat{\theta})\le
\inf_{p+1\leq n/2}  \Biggl\{ \min_{\theta\in\Theta_{p }} R( {\theta
} ) +
C \sqrt{\frac{p}{n}} \log^{5/2}(n)  \Biggr\} + C
\frac{\log\fracd{1}{\varepsilon}}{\sqrt{n}}.
\]
\end{cor}

Let us detail two examples: $\operatorname{AR}(p_{0})$ and $\operatorname{AR}(\infty)$ models with
innovations $\xi_t$
i.i.d. satisfying the Cramer condition and $\operatorname{med} (\xi_0) = 0$.

First, consider $(X_{t})$ a causal $\operatorname{AR}(p_0)$ process ($0\le p_0<\infty$)
\[
X_{t} = a_{0} + \sum_{j=1}^{p_{0}} a_{j} X_{t-j} + \xi_{j}  \qquad \mbox{for all }t\in\Z.
\]
If $B\ge\sum_{j=0}^{p}|a_{j}|+e/(p+1)$ for all $0\le p\le p_0$, the
error of the best
linear predictor is $\mu[|\varepsilon_{j}|]$.
Corollary \ref{cor1} implies, for any $0<\varepsilon<1$ and any $n\ge
2(p_0+1)$, the relation:
\[
R(\hat{\theta})-\mu[|\varepsilon_{0}|]\le
C \Biggl( \sqrt{\frac{p_{0}}{n}} \log^{5/2}(n) + \frac{\log\fracd
{1}{\varepsilon}}{\sqrt{n}} \Biggr)  \qquad \mbox{with probability at
least }1-\varepsilon.
\]
For $\varepsilon>0$ fixed independently of $n$, the rate of
convergence of the excess risk is estimated by $\sqrt{ {p_{0}}/{n}}
\log^{5/2}(n)$. Note that $\hat\theta$ achieves this rate even if
$p_0$ is unknown. One says that our procedure is adaptive in $p_0$ and,
using the terminology of \cite{modha}, memory-universal.

Second, consider $(X_t)$ a causal $\operatorname{AR}(\infty)$ process
%
\begin{equation}\label{arinf}
X_t=a_0+\sum_{i=1}^\infty a_iX_{t-i}+\xi_t \qquad \mbox{for all }t\in\Z.
\end{equation}
If
$B\ge\sum_{j=0}^{p}|a_{j}|+e/(p+1)$ for all $p\ge0$, we have
$\overline{\theta}_{p}=(a_{0},\ldots,a_{p})$. Then we roughly bound
$R(\overline{\theta}_{p})= \pi_{0}[|\sum_{i>p}a_{i}X_{-i}+\xi
_{0}|]\leq
\mu[|\xi_{0}|]+\pi_{0}[|X_{-i}|]\sum_{i>p}|a_{i}|$ and with
probability at least $1-\varepsilon$:
\[
R(\hat{\theta})-\mu[|\xi_{0}|]\le
\inf_{p+1\leq n/2}  \Biggl[
\pi_{0}[|X_{0}|]\sum_{i>p}|a_{i}| +
C \sqrt{\frac{p}{n}} \log^{5/2}(n) \Biggr] + C \frac{\log\fracd
{1}{\varepsilon}}{\sqrt{n}}.
\]
In this nonparametric setting, to obtain a rate of convergence for the
excess risk we have to specify the decay rate of the $|a_i|$. For
example, if
\[
\exists\gamma>0,\exists\beta>0, \forall p{:} \qquad
\sum_{i>p}|a_{i}| \leq\frac{\gamma}{p^{\beta}}
\]
then the convergence rate is $({\log^{5}(n)}/{n})^{\fracc{\beta
}{2\beta+1}}$ (consider the optimal $p= n^{1/(2\beta+1)}\*\log
^{5/(2\beta+1)}(n)$).

\subsubsection*{Simulations}
We implement our linear prediction procedure using the R software \cite
{R}. We compare the results
to the one obtained using the standard ARIMA procedure of R with the
AIC criterion for model selection.
Our theoretical penalization terms, driven by ``the worst-case type''
bounds, are necessarily pessimistic: our procedure systematically
over-penalizes large models. Thus, for having an efficient procedure in
practice, adjustments have been done. However, we aim with these
simulations to
show that
\begin{enumerate}[(2)]
\item[(1)] our linear prediction procedure is easily implementable;
\item[(2)] its performances are reasonable when the implemented penalization
term is smaller than the theoretical one.
\end{enumerate}
We only consider observations from simulations of $\operatorname{AR}(p_0)$ models of
the form
\[
X_{t} = \sum_{i=1}^{p_0}a_{i}X_{t-i} + \xi_{t},
\]
where the $\xi_{t}$ are i.i.d., either $\mathcal{N}(0,\sigma
^{2})$-distributed, either
$(\delta_{0}+\mathcal{E}(\lambda))/2$ distributed, where $\delta
_{0}$ is the Dirac mass
on $0$ and $\mathcal{E}(\lambda)$ the exponential distribution with
parameter $\lambda>0$. In both cases the Cramer condition is satisfied
and $\operatorname{med} (\xi_0)=0$. Unlike the first case, mean and median are
different in the second case.
Thus, the minimizers of the $\ell_1$ and the quadratic risks are the
same in the first case
and differ in the second one.

We use $p_0=3$, $a_{1}=0.2$, $a_{2}=0.3$, $a_{3}=0.2$, $\sigma^{2}\in
\{1,3\}$,
$\lambda\in\{1,1/\sqrt{12}\}$, and $n=500$,
\[
\Theta= \bigcup_{p=1}^{8} \Theta_{p} = \bigcup_{p=1}^{8}\{\theta
\in\mathds{R}^{p}\dvt\|\theta\|_{1} \leq1\}
\]
and
\[
\lambda\in\mathcal{G} = \{2,4,8,\ldots,1024\} .
\]
In view of our procedure, we compute the simplified penalized criterion
\[
(\hat{\lambda},\hat{p}) = \arg\mathop{\mathop{\min}_{
1\leq p \leq8}}_{
\lambda\in\mathcal{G}
}
- \frac{1}{\lambda}\log\int_{\Theta_{p,\ell}} \exp
(-\lambda r_n(\theta) )\,\mathrm{d}\pi_{p,\ell}(\theta)
+ \lambda\frac{K^{2}}{n}.
\]
The theoretical value $K=2(\log n)^{3/2}\approx9$ systematically
over-penalizes the large models and always selects the simplest one
($p=1$). Thus, we fix in practice $K=0,1$. To compute the criteria, the
integrand term is approximated using an acceptation-reject algorithm
with gaussian proposal and 10,000 iterations. To compare one simulation
of $\hat{\theta}\sim\pi_{\hat p}\{-\hat\lambda r_n\}$ with $\hat
{\theta}_{\mathrm{AIC}}$ obtained by the classical R procedure, we simulate
independently
$(X'_{1},\ldots,X'_{500})$ distributed as $(X_{1},\ldots, X_{500})$
and we compare
\[
\operatorname{err}_{1}(\hat{\theta}) = \frac{1}{n-8}\sum_{i=9}^{500}  \Biggl|X'_{i}
- \sum_{p=1}^{\hat{p}} (\hat{\theta})_{p} X'_{i-p}  \Biggr|
\]
with $\operatorname{err}_{1}(\hat{\theta}_{\mathrm{AIC}})$. As the classical R
procedure is based
on least square estimators, we also compare the quadratic prevision error
\[
\operatorname{err}_{2}(\hat{\theta}) = \frac{1}{n-8}\sum_{i=9}^{500} \Biggl (X'_{i}
- \sum_{p=1}^{\hat{p}} (\hat{\theta})_{p} X'_{i-p}  \Biggr)^{2}
\]
with $\operatorname{err}_{2}(\hat{\theta}_{\mathrm{AIC}})$. The results of 20
experiments are reported in
Table \ref{groszizi}.

\begin{table*}[b]
\tabcolsep=0pt
\caption{For each experiment, we report the median, mean and
standard deviation of the $\operatorname{err}_{i}(\cdot)$ quantities on the 20
experiments realized. The best results, for
both $\operatorname{err}_{1}(\cdot)$ and $\operatorname{err}_{2}(\cdot)$, are bolded for each
serie}
\label{groszizi}
\begin{tabular*}{\textwidth}{@{\extracolsep{\fill}}llllll@{}}
\hline
$\xi_{t}$ & & $\operatorname{err}_{1}(\hat{\theta})$ & $\operatorname{err}_{1}(\hat
{\theta}_{\mathrm{AIC}})$
& $\operatorname{err}_{2}(\hat{\theta})$ & $\operatorname{err}_{2}(\hat{\theta
}_{\mathrm{AIC}})$ \\
\hline
$\mathcal{N}(0,1)$ & median & {\bf0.790} & 0.792 & {\bf0.975} & {\bf
0.975} \\
& mean & {\bf0.797} & 0.798 & {\bf0.985} & 0.988 \\
& s.d. & 0.023 & 0.024 & 0.054 & 0.054 \\
[3pt]
$\mathcal{N}(0,3)$ & median & 2.433 & {\bf2.432} & 0.918 & {\bf
0.916} \\
& mean & {\bf2.409} & 2.412 & {\bf0.911} & 0.912 \\
& s.d. & 0.078 & 0.065 & 0.496 & 0.412 \\
[3pt]
$\frac{\delta_{0}+\mathcal{E}(1)}{2}$ & median & {\bf0.567} & 0.592
& 0.819 & {\bf0.813} \\
& mean & {\bf0.580} & 0.589 & 0.836 & {\bf0.813} \\
& s.d. & 0.047 & 0.043 & 0.153 & 0.150 \\
[3pt]
$\frac{\delta_{0}+\mathcal{E}(1/\sqrt{12})}{2}$ & median & {\bf
1.973} & 2.000 & 9.525 & {\bf9.494} \\
& mean & {\bf1.955} & 1.997 & 9.733 & {\bf9.390} \\
& s.d. & 0.158 & 0.162 & 1.656 & 1.522 \\
\hline
\end{tabular*}
\end{table*}

The results are coherent with the theory: in the Gaussian cases, the
optimal values of $\theta$
for the $\ell_{1}$ and the quadratic risks of prediction are the same.
Both procedures estimate efficiently the same $\theta$
and their prediction risks are the same.
In the other cases, the optimal values of $ \theta$ for the $\ell
_{1}$ and the
quadratic risks are not the same. We observe
$\operatorname{err}_{1}(\hat{\theta})<\operatorname{err}_{1}(\hat{\theta}_{\mathrm{AIC}})$ and
$\operatorname{err}_{2}(\hat{\theta})>\operatorname{err}_{2}(\hat{\theta}_{\mathrm{AIC}})$.
The choice between
the two procedures only depends on the prediction risk considered.

\subsection{Neural networks predictors}\label{app::NN}
Similarly than in \cite{modha}, we present a procedure that
approximates the best possible predictor using the best possible number
of past values $p$ for the one-step prediction. Given $p$, the best
possible predictor for the $\mathbb L^1$-risk is $\operatorname
{med}(X_0|X_{-1},\ldots,X_{-p})$. We denote $R_p^\ast$ the
corresponding risk. For $\mathcal X=\R$, we use the abstract neural
networks predictors
defined in Barron \cite{Barron1994} by the relation
\[
f_\theta= c_0+\sum_{i=1}^\ell c_i  \phi(a_i\cdot x+b_i) \qquad  \mbox{for all }x\in\R^p
\]
for $a_i\in\R^p$ and $c_i,b_i\in\R$ for all $1\le i\le\ell$, the
sigmoidal function $\phi(x)=(1+\exp(-x))^{-1}$ for all $x\in\R$ and
$\theta=(c_0,a_{1,1},\ldots,a_{1,p},b_{1},c_{1},\ldots,
a_{\ell,1},\ldots,a_{\ell,p},b_{\ell},c_{\ell})$ in $\mathcal
B^q_{c_{p,\ell}}$ for some $c_{p,\ell}>0$, $q={\ell(p+2)+1}$ and
$\ell\le n $.
For any $p\ge1$, we denote
\[
r_p(x)=\operatorname{med}\bigl(X_0|(X_{-1},\ldots,X_{-p})=x\bigr) \qquad \mbox{for all }x\in
\R^p
\]
and we assume that there exists a complex-valued function $\tilde r_p$
on $\R^p$ satisfying
\[
\forall x\in\R^p  \qquad  r_p(x)-r_p(0)=\int_{\R^p}(\mathrm{e}^{\mathrm{i}wx}-1)\tilde r_p(w)\,\mathrm{d}w
  \quad \mbox{and} \quad
\int_{\R^p}\|w\|_1|\tilde r_p(w)|\,\mathrm{d}w\le C'p^c
\]
for some $C',c>0$. Then
\begin{cor}\label{cor::nn}
Under {({WDP})} if for any $(p,\ell)\in M$
%
\begin{equation}\label{cpl}
\frac{q}{e}+2\sqrt\ell\|X\|_\infty(C'p^c+\ell\log\ell)\le
c_{p,\ell}
\end{equation}
then, for all $n\ge\max_M c_{p,\ell}$, with probability at least
$1-\varepsilon$,
\[
R(\hat{\theta})\leq
\inf_{10(1+\log n)^2p^{1+2c} \le n }  \biggl\{
R^\ast_p +
C\frac{p^{1/4+c/2} \log^{3} n}{n^{1/4}}  \biggr\}+ C \frac{\log
\fracd{1}{\varepsilon}}{\sqrt{n}}.
\]
\end{cor}

If $(X_t)$ satisfies the Markov condition of order $p_0$, then $c=0$
and for $n$ sufficiently large
\[
R(\hat{\theta})-R^\ast_{p_0}\leq C \biggl(\frac{ \log^{3}
n}{n^{1/4}} + \frac{\log\fracd{1}{\varepsilon}}{\sqrt{n}}\biggr ).
\]
Compared to the i.i.d. case, the loss is $\log^3 n$ and we do not need
to know the order $p_0$ (our procedure is memory-universal). Our loss
is smaller than the one of the other memory-universal procedure given in
\cite{modha}.

\subsection{Nonparametric auto-regressive predictors}\label{app::np}

As in Baraud, Comte and Viennet \cite{Baraud}, we assume that
$(X_t)$ is a solution of the equation:
\[
X_{t} = f_1(X_{t-1}) +\cdots + f_{p_{0}}(X_{t-p_{0}})+ \xi_{t} \qquad \mbox{for all
}t\in\Z,
\]
where $\xi_{t}\sim\mathcal{N}(0,\sigma^{2})$, the
$f_i$ are functions $[-1;1] \mapsto\mathds{R}$ in H\"older class
$ H(s_i,L_i) $: $f_i$ is derivable $\lfloor s_i\rfloor$ times and
%
\begin{equation}
\label{holder2}
\exists\mathcal L_i>0,\forall(x,x')\in[-1,1]^{2}, \qquad
\bigl|f_i^{(\lfloor s_i\rfloor)}(x)-f_i^{(\lfloor s_i\rfloor)}(x')\bigr|
\leq\mathcal L_i|x-x'|^{s_i-\lfloor s_i\rfloor} .
\end{equation}
Consider the Fourier basis $(\phi_{j}(\cdot))_{j\geq1}$
on $[-1,1]$ composed by $\phi_{2k}(x)=\sqrt{2}\cos(2\uppi k x)$ and
$\phi_{2k+1}(x)=\sqrt{2}\sin(2\uppi k x)$.
Assumption \ref{holder2} implies the existence of $\gamma_i>0$ such
that for any $m\ge0$ it holds
\[
\min_{(\alpha_{1},\ldots,\alpha_{m})\in\mathds{R}^{m}}
 \Biggl\{\int_{-1}^{1} \Biggl [f_i(t)-\sum_{j=1}^{m}\alpha_{i,j}\phi
_{j}(t) \Biggr]^{2}\,\mathrm{d}s \Biggr\}^{\fraca{1}{2}}
\leq\gamma_{i}m^{- s_i} .
\]
Natural predictors are given by
\[
\widehat X_{n+1}=\sum_{i=1}^p\sum_{j=1}^\ell\theta_{i,j}\varphi
_j(X_{n-i})=:f_\theta(X_n,\ldots,X_{n-p})
\]
for any $p\in\{1,\ldots,\lfloor n/2\rfloor\}$ and any $\ell\in\{
1,\ldots,m_{p}=n\}$. We restrict the procedure on $\theta_{p,\ell} $ in
the compact set
\[
\Theta_{p,\ell}= \Biggl\{\theta\in\mathds{R}^{p \ell},
\sum_{i=1}^{p} \sum_{j=1}^{\ell} \theta_{i,j}^{2} (2[j/2])^{2} \leq
L^{2}  \Biggr\}
\]
such that any $f_{\theta}$ is an $L$-Lipschitz function. We define
also the coefficients $\overline\theta_{p,\ell}\in\R^{p\ell}$ by
the relation
\[
\arg\min_{\theta\in\Theta_{p,\ell}}\pi_0\Biggl [ \Biggl|X_n-\sum
_{i=1}^p\sum_{j=1}^\ell\theta_{i,j}\varphi_j(X_{n-i}) \Biggr| \Biggr].
\]
As a consequence of Theorem \ref{mainthm}, it holds
\begin{cor}
\label{fourier}
Under {(CBS)}, if for any $\ell\ge1$ and any $p\ge1$
\[
\frac{\ell p}{e} +  \Biggl(\sum_{i=1}^{p_{0}} \sum_{j=1}^{\ell}
(\overline{\theta}_{p_{0},\ell})_{i,j}^{2} (2\lfloor j/2\rfloor)^{2}
 \Biggr)^{\fraca{1}{2}}\le L,
\]
then for all $n\ge8e (1+L)$ with probability at least $1-\varepsilon$
\[
R(\hat{\theta})-\mu[|\xi_{0}|]
\leq C  \biggl( \biggl(\frac{\log(n)}{n} \biggr)^{\fracc{s}{2s+1}} +
\frac{\log\fracd{1}{\varepsilon}}{\sqrt{n}} \biggr),
\]
where $s$ denotes $\min\{ s_{1},\ldots, s_{p_{0}}\}$.
\end{cor}

The (i.i.d.) minimax rate of convergence with respect to $ s_{1},\ldots, s_{p_{0}}$
for the $\ell^1$-risk is achieved up
to a logarithmic loss.
In \cite{Baraud}, the (i.i.d.) minimax rate of convergence for the
quadratic risk
is achieved for the empirical quadratic risk in expectation.

\section{Proofs}\label{sec::pf}

To present the proofs in a unified version whether we work under {(CBS)} or {(WDP)}, we truncate the observations if we are under
{(CBS)}:
\[
\overline X_t=H(\overline\xi_t,\overline\xi_{t-1},\overline\xi
_{t-2}, \ldots) \qquad \mbox{for all }t\in\Z,
\]
where $\overline\xi_t=(\xi_t\wedge C)\vee
(-C)$,
under {(WDP)} we just take $\overline X_t=X_t$.
We denote in the sequel $\overline X=(\overline X_t)$ and $\overline
r$, $\overline R$ the risks
associated with $\overline X$ under {(CBS)} and with $X$ under {(WDP)}. To shorten the proofs, we denote $K_n=(1+L)\log^{3/2} n$ and
$w_{p,\ell}=1/(m_{p} \lfloor n/2\rfloor)$ in the sequel.
The proof of our main theorem lies on estimates on Laplace transforms.

\subsection{Preliminary lemmas: Estimates on Laplace transforms}

The proofs of these lemmas are given in Section \ref{prooflems}.
The first lemma is an estimate of the Laplace transforms of the risk of
$\overline X$; it is
a direct corollary of the result in Rio \cite{Rio2000a}.
\begin{lemma}[(Laplace transform of the risk)] \label{lemrisk}
For any $\lambda>0$ and $\theta\in\Theta$ we have:
\[
\pi_0 \bigl[\exp\bigl( \lambda\bigl(\overline R(\theta) - \overline r_n(\theta)\bigr) \bigr)\bigr]
\leq
\exp \biggl(\frac{ \lambda^{2}k_n^{2}}{ n(1-p /n)^{2}} \biggr),
\]
where $k_n=\sqrt2C(1+L) (a(H)+\tilde a(H))$ under {(CBS)} and
$k_n= (1+L) (\|X_0\|_\infty+\theta_{\infty,n}(1))/\sqrt2$ under
{({WDP})}.
\end{lemma}

Given a measurable space $(E,\mathcal{E})$ we let
$\mathcal{M}_{+}^{1}(E)$ denote the set of all probability measures
on $(E,\mathcal{E})$. The Kullback divergence is a pseudo-distance on
$\mathcal{M}_{+}^{1}(E)$ defined, for any
$(\pi,\pi')\in[\mathcal{M}_{+}^{1}(E)]^{2}$ by the equation
\[
\mathcal{K}(\pi,\pi')=  \cases{
\pi[\log(\mathrm{d} \pi/\mathrm{d} \pi')],& \quad     if   $\pi\ll\pi'$,
\cr
+ \infty, & \quad   otherwise.
}
\]
The proof of the following lemma is omitted as it can be found in \cite
{Classif} or \cite{Cat7}.
\begin{lemma}[(Legendre transform of the Kullback divergence function)]
\label{LEGENDRE}
For any $\pi\in\mathcal{M}_{+}^{1}(E)$,
for any measurable function $h\dvtx E\rightarrow\mathds{R}$ such
that $ \pi[\exp(h)]<+\infty$ we have:
%
\begin{equation} \label{lemmacatoni}
\pi[\exp
(h)]=\exp \Bigl(\sup_{\rho\in\mathcal{M}_{+}^{1}(E)} \bigl(\rho
[h]-\mathcal{K}(\rho,\pi) \bigr) \Bigr),
\end{equation}
with convention $\infty-\infty=-\infty$. Moreover, as soon as $h$ is
upper-bounded on the support of $\pi$, the supremum with respect to
$\rho$ in the right-hand side is reached for the Gibbs measure
$\pi\{h\}$ defined in \eqref{Gibbsmeasure}.
\end{lemma}

Using Lemmas \ref{lemrisk} and \ref{LEGENDRE}, we get an upper-bound for
the Laplace transform of the mean risk of Gibbs estimators in all sub-models.
\begin{lemma} \label{lemrisk2}
Under the
assumptions of Theorem \ref{mainthm}, we have for any
$\lambda>0$ and $(p,\ell)\in M$:
%
\begin{equation}
\label{equlem1,equlem2}
\pi_0 \biggl[ \exp \biggl( \sup_{\rho\in\mathcal{M}_{+}^{1}(\Theta
_{p,\ell})}
 \{ \lambda\rho[\overline R- \overline r_n]
-\mathcal{K} (\rho,\pi_{p,\ell} )  \}
  - \frac{\lambda^{2}k_n^{2}}{n(1-p/n)^{2}} \biggr) \biggr]\leq1,
\end{equation}
where $k_n$ has the same expression than in Lemma \ref{lemrisk}.
\end{lemma}

Following the technique used by Catoni \cite{Classif}, we derive from Lemma
\ref{lemrisk2} another upper-bound on the Laplace transform of the
mean risk
of any aggregation estimators of all Gibbs estimators.
\begin{lemma}
\label{lemrisk3}
For any measurable functions
$ \hat\rho_{p,\ell}\dvtx
\mathcal{X}^{n} \rightarrow\mathcal{M}_{+}^{1}(\Theta_{p,\ell})$
for $(p,\ell)\in M$,
under the assumptions of Theorem \ref{mainthm}, we have:
\[
\pi_0 \biggl[ \sum_{(p,\ell)\in M} \sum_{\lambda\in\mathcal
{G}_{p,\ell}}
\hat\rho_{p,\ell} \biggl[\exp \biggl(
\lambda(\overline R- \overline r_n)
- \log\frac{\mathrm{d}\hat\rho_{p,\ell}}{\mathrm{d}\pi_{p,\ell}}-\frac{\lambda
^{2} k_n^{2}}{n(1-p/n)^{2}}
+\log(w_{p,\ell}/n)  \biggr) \biggr] \biggr]\le1
\]
and
\[
 \pi_0 \biggl[
\sum_{(p,\ell)\in M} \sum_{\lambda\in\mathcal{G}_{p,\ell}}
\exp \biggl(
\lambda \hat\rho_{p,\ell} [\overline r_{n} - \overline R]
- \mathcal{K}(\hat\rho_{p,\ell},\pi_{p,\ell})-\frac{\lambda^{2}
k_n^{2}}{n(1-p/n)^{2}}
+\log(w_{p,\ell}/n) \biggr) \biggr] \le1,
\]
where we remind that $k_n$ is defined in Lemma \ref{lemrisk}.
\end{lemma}

Finally, we use a lemma that quantify the error in the risk due to the
truncation under {(CBS)}.
\begin{lemma}
\label{lembound}
Under {(CBS)}, for any truncation level $C>0$ and any $0\le
\lambda\le n/(4(1+L))$, we have
\[
\pi_0\biggl [\exp\biggl( \lambda\sup_{\theta\in\Theta}|r_n(\theta
)-\overline r_n(\theta)|
-\lambda2 (1+L)\Psi(a(H)) \biggl(\frac{a(H)^2C}{\exp(a(H)C)-1}
+\lambda\frac{4(1+L)}{n}  \biggr)\biggr) \biggr]\le1.
 \]
\end{lemma}

\subsection{Proof of Theorem \protect\ref{mainthm}}
\label{proofmain}
Remark that {(WDP)} is satisfied, so $\overline R=R$ and $\overline
r=r$. We apply the first
inequality of Lemma~\ref{lemrisk3} to $\hat\rho_{p,\ell}^\lambda
=\pi_{p,\ell}\{-\lambda r_n\}$.
Remembering that $(\hat p,\hat\ell, \hat\lambda)= \arg\min\hat
{R} (p,\ell,\lambda )$,
we obtain in particular:\looseness=1
%
\begin{equation}\label{lastlapl}
\pi_0 \hat\rho_{\hat p,\hat\ell}^{\hat\lambda}\biggl [\exp \biggl(
\hat\lambda(R- r_n)
- \log \biggl(\frac{\mathrm{d}\hat\rho_{\hat p,\hat\ell}^{\hat\lambda}}
{\mathrm{d}\pi_{\hat p,\hat\ell}} \biggr)-\frac{\hat\lambda^{2}
k_n^{2}}{n(1-\hat p/n)^{2}}
+\log \biggl(\frac{w_{\hat p,\hat\ell}}{n} \biggr) \biggr)
\biggr]\le1.
\end{equation}
Remark that $\pi_0   \hat\rho_{\hat p,\hat\ell}^{\hat\lambda}$
is a well defined probability
measure as $\hat\rho$ are defined conditionally on the observations.
Remark also that $\hat\theta
\sim \hat\rho_{\hat p,\hat\ell}^{\hat\lambda}$ by definition,
then using the classical Chernov
bound we derive that with probability $1-\varepsilon$ it holds:
%
\begin{equation}
\label{equproof1}
R(\hat{\theta})
\leq
r_n (\hat\theta)
+ \frac{\hat\lambda k_n^{2}}{n(1-\hat p/n)^{2}}
+ \frac{1}{\hat\lambda}\log \biggl( \frac{\mathrm{d}\hat\rho_{\hat p,\hat
\ell}^{\hat\lambda}}{\mathrm{d}\pi_{\hat p,\hat\ell}} \biggr)
+ \frac{1}{\hat\lambda} \log \biggl( \frac{ n}{w_{\hat p,\hat\ell
}} \biggr)
+ \frac{1}{\hat\lambda} \log\frac{1}{\varepsilon}.
\end{equation}
In order that the term $\hat R$ appears, we notice that \eqref
{equproof1} is equivalent to
\begin{eqnarray*}
  R(\hat{\theta})
& \leq&
- \frac{1}{\hat\lambda} \log\int_{\Theta_{\hat p,\hat\ell}}
\exp (-\hat\lambda r_n(\theta) )\pi_{\hat p,\hat\ell
}(\mathrm{d}\theta) +
\frac{\hat\lambda k_n^{2}}{n(1-\hat p/n)^{2}} + \frac{1}{\hat
\lambda} \log
 \biggl(\frac{n}{ w_{\hat p,\hat\ell}} \biggr) +
\frac{1}{\hat\lambda} \log\frac{1}{\varepsilon}
 \\
& \le&\inf_{p,\ell,\lambda}
\hat{R} (p,\ell,\lambda )
+\frac{\hat\lambda(k_n^2-K_n^2)}{n(1-\hat p/n)^2}- \frac{1}{\hat
\lambda} \log\varepsilon
\end{eqnarray*}
(remind that $K_n=(1+L)\log^{3/2} n$).
Now, we upper bound the term
$\hat{R} (p,\ell,\lambda )$, for any $p$, $\ell$ and
$\lambda$.
Using the second inequality of Lemma
\ref{lemrisk3}, we obtain for any $(p,\ell)\in M$, $\lambda\in
\mathcal{G}$
and $\rho\in\mathcal{M}_{+}^{1}(\Theta_{p,\ell})$,
%
\begin{eqnarray}
\label{equproof1prime} \int_{\Theta_{p,\ell}} r_n
 (\theta ) \rho(\mathrm{d}\theta) &\leq&\int_{\Theta_{p,\ell}} R
 (\theta ) \rho(\mathrm{d}\theta) + \frac{\lambda
k_n^{2}}{n(1-p/n)^{2}} + \frac{1}{\lambda}
\mathcal{K} (\rho,\pi_{p,\ell} ) \nonumber
\\[-8pt]
\\[-8pt]&&{}+ \frac{1}{\lambda}
\log
\frac{n}{w_{p,\ell}} + \frac{1}{\lambda}
\log\frac{1}{\varepsilon}.
\nonumber
\end{eqnarray}
From \eqref{equproof1prime} and using Lemma \ref{LEGENDRE} two times,
we derive that
\begin{eqnarray*}
&&- \frac{1}{\lambda} \log\int_{\Theta_{p,\ell}} \exp
(-\lambda r_n(\theta) )\pi_{p,\ell}(\mathrm{d}\theta)\\
&& \quad =
\inf_{\rho\in\mathcal{M}_{+}^{1}(\Theta_{p,\ell})}
 \biggl\{\int_{\Theta_{p,\ell}} r_n  (\theta ) \rho
(\mathrm{d}\theta)
+ \frac{1}{\lambda} \mathcal{K} (\rho,\pi_{p,\ell} )
 \biggr\}\\
&& \quad \leq
\inf_{\rho\in\mathcal{M}_{+}^{1}(\Theta_{p,\ell})}
 \biggl\{\int_{\Theta_{p,\ell}} R  (\theta ) \rho
(\mathrm{d}\theta)
+ \frac{2}{\lambda} \mathcal{K} (\rho,\pi_{p,\ell} )
 \biggr\}
+ \frac{\lambda k_n^{2}}{n(1-p/n)^{2}}
+\frac{1}{\lambda} \log\frac{n}{\varepsilon w_{p,\ell}}
\\
&& \quad =
- \frac{2}{\lambda} \log\int_{\Theta_{p,\ell}} \exp \biggl(-\frac
{\lambda}{2} R(\theta) \biggr)\pi_{p,\ell}(\mathrm{d}\theta)
+\frac{\lambda k_n^{2}}{n(1-p/n)^{2}}
+\frac{1}{\lambda} \log\frac{n}{\varepsilon w_{p,\ell}}.
\end{eqnarray*}
Finally, we obtain:
%
\begin{equation}\label{equproof5}
\hat{R} (p,\ell,\lambda )
\leq
- \frac{2}{\lambda} \log\int_{\Theta_{p,\ell}} \exp \biggl(-\frac
{\lambda}{2} R(\theta) \biggr)\pi_{p,\ell}(\mathrm{d}\theta)
+\frac{\lambda( k_n^{2}+K_n^2)}{n(1-p/n)^{2}}
+\frac{1}{\lambda} \log\frac{n}{ \varepsilon w_{p,\ell}}.
\end{equation}
Under Assumption \eqref{dimension}, as soon as $\lambda>2e$ it holds
\[
- \log\pi_{p,\ell} \biggl[\exp \biggl(-\frac{\lambda}{2}
 \bigl( R- R(\overline{\theta}_{p,\ell}) \bigr) \biggr) \biggr]
\le d_{p,\ell}\log\frac{\lambda}{2}
\]
and it easily follows that
\[
-\log\pi_{p,\ell} \biggl[\exp \biggl(-\frac{\lambda}{2}\overline
R \biggr) \biggr]
\leq d_{p,\ell} \log\frac{\lambda}{2} + \frac{\lambda}{2}
R(\overline{\theta}_{p,\ell}).
\]
We plug this result into the inequality \eqref{equproof5} to obtain:
%
\begin{equation}
\label{equproof6}
\hat{R} (p,\ell,\lambda )
\leq
R (\overline{\theta}_{p,\ell} ) +\frac{1}{\lambda}
\biggl(2 d_{p,\ell} \log
\frac{\lambda}{2}+\log\frac{n}{\varepsilon w_{p,\ell}} \biggr)
+\frac{\lambda(k_n^{2}+K_n^2)}{n(1-p/n)^{2}}.
\end{equation}
Collecting the inequalities \eqref{equproof1} and \eqref{equproof6},
we obtain:
%
\begin{eqnarray}
\label{equproof6-2}
R(\hat{\theta})
&\le&
\inf_{p,\ell,\lambda\in\mathcal{G}_{p,\ell}}
 \biggl\{R (\overline{\theta}_{p,\ell} ) + \frac
{1}{\lambda} \biggl( 2 d_{p,\ell}
\log\frac{\lambda}{2}+ \log\frac{n}{\varepsilon w_{p,\ell}}
 \biggr)
+\frac{\lambda(k_n^{2}+K_n^2)}{n(1-p/n)^{2}} \biggr\}\nonumber
\\[-8pt]
\\[-8pt]
&&{}+
\frac{ \hat\lambda(k_n^2-K_n^2)}{n(1-\hat p/n)^2}- \frac{1}{\hat
\lambda} \log\varepsilon.
\nonumber
\end{eqnarray}
As for $\lambda\in\mathcal{G}_{p,\ell}$, we have, by definition of
$\mathcal{G}_{p,\ell}$ that
\[
\lambda\in \biggl[\check{c} \frac{\sqrt{d_{p,\ell}n}\log
(d_{p,\ell}n)}{K_{n}},\ldots,
\hat{c} \frac{\sqrt{d_{p,\ell}n}\log(d_{p,\ell}n)}{K_{n}} \biggr]
\cap[2e,n]
\]
then it holds
%
\begin{eqnarray}
\label{equproof6-2-bis}
R(\hat{\theta})
&\le&
\inf_{d_{p,\ell}\leq n}
\biggl \{R (\overline{\theta}_{p,\ell} ) +
\frac{K_{n}}{\check{c}\sqrt{d_{p,\ell}n}\log(d_{p,\ell}n)}
 \biggl( 2 d_{p,\ell}
\log\frac{n}{2}+ \log\frac{n}{\varepsilon w_{p,\ell}}
 \biggr)
\nonumber
\\[-8pt]
\\[-8pt]
&&\hphantom{\inf_{d_{p,\ell}\leq n}
\biggl \{}{}+4 \hat{c} (k_n^{2}+K_n^2) \sqrt{\frac{d_{p,\ell}}{n}} \frac{\log
(n d_{p,\ell})}{K_n}  \biggr\}
+
4 (k_{n}^{2}-K_{n}^{2})_{+} + \frac{(1+L)\log\fracd{1}{\varepsilon
}}{\check{c} \sqrt{n}}.
\nonumber
\end{eqnarray}
For the sake of simplicity, we use rough estimates ($1\leq d_{p,\ell}$,
$1\leq {1}/{\varepsilon}$, $m_p\le n$, \ldots) to obtain
\begin{eqnarray*}
R(\hat{\theta})
&\le&
\inf_{d_{p,\ell}\leq n}
 \Biggl\{R (\overline{\theta}_{p,\ell} ) +
(1+L)  \biggl( \frac{6}{\check{c}}
+ 8 \hat{c}  \bigl(1+\|X_{0}\|_{\infty} + \theta_{\infty,n}(1) \bigr)^{2}
 \biggr) \sqrt{\frac{d_{p,\ell}}{n}} \log^{5/2} (n)
 \Biggr\}
\\
&&{}+
4 (k_{n}^{2}-K_{n}^{2})_{+} + \frac{7(1+L)\log\fracd{n}{\varepsilon
}}{\check{c} \sqrt{n}}.
\end{eqnarray*}
This ends the proof as
\[
k_{n}^{2}-K_{n}^{2}=(1+L)  \biggl(\frac{(\|X_{0}\|_{\infty}+\theta
_{\infty,n}(1))^{2}}{2} - \log^{3}(n) \biggr).
\]

\subsection{Proof of Theorem \protect\ref{mainthm2}}
\label{proofmain2}
As we work under {(CBS)}, we have to deal with the error of
approximation of $r$ and $R$ by $\overline R$. To quantify it, we use
Lemma \ref{lembound}. First, remark that as $R=\pi_0[r]$ it holds
\[
\exp \Bigl( \lambda\sup_{\theta\in\Theta}|R(\theta)-\overline
R(\theta)|-\lambda\phi(C,\lambda) \Bigr)\le1,
\]
where
\[
\phi(C,\lambda)= 2 (1+L)\Psi(a(H)) \biggl(\frac{a(H)^2C}{\exp(a(H)C)-1}
+\lambda\frac{4(1+L)}{n}  \biggr).
\]
An immediate consequence is that
\[
\pi_0 \Bigl[\exp \Bigl( \lambda\sup_{\theta\in\Theta}|
(r_n-R)(\theta)-(\overline r_n-\overline R)(\theta)|-2\lambda\phi
(C,\lambda) \Bigr) \Bigr]\le1.
\]
As $R-r_n=\overline r_n-\overline R+(r_n-R) -(\overline r_n-\overline
R)$, for any measurable function
$ \rho_{p,\ell}\dvtx
\mathcal{X}^{n} \rightarrow\mathcal{M}_{+}^{1}(\Theta_{p,\ell})$
the Cauchy--Schwarz inequality gives
\begin{eqnarray*}
&&\pi_0\rho\bigl[\exp\bigl(\lambda/2(R-r_n)\bigr)\bigr]\\
&& \quad \le\sqrt{\pi_0\rho\bigl[\exp
\bigl(\lambda(\overline R-\overline r_n)\bigr)\bigr]
\pi_0\rho \Bigl[\exp \Bigl(\lambda\sup_{\theta\in\Theta}|
(r_n-R)(\theta)-(\overline r_n-\overline R)(\theta)| \Bigr) \Bigr]}.
\end{eqnarray*}
Using this remark and the same reasoning than in the proof of Theorem
\ref{mainthm} that gives
\eqref{lastlapl} from Lemma \ref{lemrisk3}, we get the inequality
\begin{eqnarray*}
&&\pi_0 \hat\rho_{\hat p,\hat\ell}^{\hat\lambda}\biggl [\exp \biggl(
\frac{\hat\lambda}{2} (R- r_n)
- 0,5\log \biggl(\frac{\mathrm{d}\hat\rho_{\hat p,\hat\ell}^{\hat\lambda
}}{\mathrm{d}\pi_{\hat p,\hat\ell}} \biggr)-0,5\frac{\hat\lambda^{2}
k_n^{2}}{n(1-\hat p/n)^{2}}
\\
&&\hphantom{\pi_0 \hat\rho_{\hat p,\hat\ell}^{\hat\lambda}\biggl [\exp \biggl(}{}+0,5\log \biggl(\frac{w_{\hat p,\hat\ell}}{n} \biggr)
   -\lambda\phi(C,\lambda) \biggr) \biggr]\le1.
\end{eqnarray*}
As in the proof of Theorem \ref{mainthm}, we derive an equivalent of
\eqref{equproof1}, that is, with probability $1-\varepsilon$ it holds:
\[
R(\hat{\theta})
\leq
r_n (\hat\theta)
+ \frac{\hat\lambda k_n^{2}}{n(1-\hat p/n)^{2}}
+ \frac{1}{\hat\lambda}\log \biggl( \frac{\mathrm{d}\hat\rho_{\hat p,\hat
\ell}^{\hat\lambda}}{\mathrm{d}\pi_{\hat p,\hat\ell}} \biggr)
+ \frac{1}{\hat\lambda} \log \biggl( \frac{ n}{w_{\hat p,\hat\ell
}} \biggr)
+2\phi(C,\hat\lambda)+ \frac{2}{\hat\lambda} \log\frac
{1}{\varepsilon}.
\]
With similar arguments, we derive an equivalent of \eqref{equproof1prime}:
\begin{eqnarray*}
\int_{\Theta_{p,\ell}} r_n
 (\theta ) \rho(\mathrm{d}\theta) &\leq&\int_{\Theta_{p,\ell}} R
 (\theta ) \rho(\mathrm{d}\theta) + \frac{\lambda
k_n^{2}}{n(1-p/n)^{2}} + \frac{1}{\lambda}
\mathcal{K} (\rho,\pi_{p,\ell} ) + \frac{1}{\lambda}
\log
\frac{n}{w_{p,\ell}}
\\
&&{}+2\phi(C,\lambda)+ \frac{2}{ \lambda} \log\frac{1}{\varepsilon}\\[-15pt]
\end{eqnarray*}
and also\vspace*{-1pt}
%
\begin{eqnarray}
\label{preuve2-pierre}
R(\hat{\theta})
&\le&
\inf_{p,\ell,\lambda} \biggl\{R (\overline{\theta}_{p,\ell
} ) +
\frac{1}{\lambda} \biggl( 2 d_{p,\ell} \log\frac{\lambda}{2}+ \log
\frac{n}{\varepsilon w_{p,\ell}}
 \biggr)
+\frac{\lambda(k_n^2+K_n^{2})}{n(1-p/n)^{2}}
+ 2\phi(C,\lambda) \biggr\}\nonumber
\\[-8pt]
\\[-8pt]
&&{}+\frac{\hat{\lambda} (k_n^2-K_n^{2})}{n(1-p/n)^{2}}
+2\phi(C,\hat\lambda)
- \frac{2}{\hat\lambda} \log\varepsilon.
\nonumber
\end{eqnarray}
We still have\vspace*{-1pt}
\[
-\frac{2}{\hat\lambda} \log\varepsilon\leq\frac{2(1+L)}{\check
{c}\sqrt{n}} \log\frac{1}{\varepsilon}
\]
so we now have to upper bound $ 2\phi(C,\hat\lambda) $. As $\hat
\lambda\leq n/(4(1+L))$ by definition of the $\mathcal G_{p,\ell}$,
fixing $C=a(H)^{-1} \log n $ we obtain:\vspace*{-1pt}
\[
\phi(C,\hat\lambda)\le\frac{4a(H)(1+L)\Psi(a(H)) [2\hat\lambda
(1+L) +a(H) \log(n)]}{n}.
\]
As $\hat\lambda\le\hat c{d_{\hat{p},\hat{\ell}}}\log(d_{\hat
{p},\hat{\ell}}n)/(1+L)$ by definition of $\mathcal G_{p,\ell}$, we obtain
\begin{eqnarray*}
\frac{\hat\lambda(k_n^2-K_n^{2})}{n(1-p/n)^{2}}
+2\phi(C,\hat\lambda)
&\leq&
\frac{8a(H)^2(1+L)\Psi(a(H))\log(n)}{n}
\\[-1.5pt]
&&{}+\sqrt{\frac{d_{\hat{p},\hat{\ell}}}{n}}\log(d_{\hat{p},\hat
{\ell}}n)4(1+L)\\
&&{}\times\hat c
 \bigl(4a(H)\Psi(a(H)) + 2\log^{2}(n)\bigl(1+\tilde{a}(H)/a(H)\bigr)^{2} -
\log^{3}(n)  \bigr)_{+}.
\end{eqnarray*}
We now plug this result into \eqref{preuve2-pierre} to end the proof.\vspace*{-1.5pt}

\subsection{Proofs of Lemmas \protect\ref{lemrisk}, \protect\ref{lemrisk2}, \protect\ref{lemrisk3}, \protect\ref{lembound} and of Proposition \protect\ref{propdim}}
\label{prooflems}
\vspace*{-3pt}
\begin{pf*}{Proof of Lemma \ref{lemrisk}}
The proof of this lemma is based on the following result of Rio \cite
{Rio2000a} on~$\overline X$.
\begin{thm}
\label{RIOthm}
Let $Y=(Y_t)_{t\in\Z}$ be a bounded stationary time series bounded
distributed as $\pi_0$ on $\mathcal{X}^\Z$. Let $h$ be a
$1$-Lipschitz function of $ \mathcal{X}^{n} \rightarrow\mathds{R} $,
that is,
such that:\vspace*{-1pt}
%
\begin{equation} \label{lip}
\forall(x_{1},y_{1},\ldots,x_{n},y_{n})\in\mathcal{X}^{2n},  \qquad  |
h(x_{1},\ldots,x_{n}) - h(y_{1},\ldots,y_{n}) |
\leq\sum_{i=1}^{n}  \|x_{i}-y_{i} \|.
\end{equation}
Then for any
$t\ge0$ we have:
\[
\pi_0 \bigl[\exp\bigl(t\bigl(\pi
_0[h(X_{1},\ldots,X_{n})]-h(X_{1},\ldots,X_{n})\bigr)\bigr) \bigr]
\le\exp\bigl(t^2n\bigl(\|X_0\|_\infty+ \theta_{\infty,n}(1) \bigr)^{2}/2 \bigr).
\]
\end{thm}

\begin{pf}
This version of Theorem 1 of \cite{Rio2000a} comes rewriting the
inequality (3) in \cite{Rio2000a} as, for any $1$-Lipschitz function $g$:
\[
\Gamma(g)=\|\E(g(X_{\ell+1},\ldots,X_n)|{\mathcal F}_\ell)-\E
(g(X_{\ell+1},\ldots,X_n))\|_\infty\le\theta_{\infty,n-\ell}(1).
\]
The result is proved as $\sup_{1\le r\le n}\theta_{\infty,r}(1)\le
\theta_{\infty,n}(1).$
\end{pf}

We now apply the result of Theorem \ref{RIOthm} on $Y=\overline X$ to
obtain the result of Lemma \ref{lemrisk}.
Let us fix $\lambda>0$, $(p,\ell)\in M$, $\theta\in\Theta_{p,\ell
}$ and $t=(1+L)\lambda/ [n-p (\theta ) ]$ and
the function $h$ defined by:
\[
h(x_{1},\ldots,x_{n})
= \frac{1}{1+L} \sum_{i=p (\theta )+1}^{n}
 \bigl\|x_{i} - f_{\theta}\bigl(x_{i-1},\ldots,x_{i-p (\theta )}\bigr)
 \bigr\|.
\]
We easily check that $h$ satisfies condition \eqref{lip}:
\begin{eqnarray*}
 &&|h(x_{1},\ldots,x_{n}) - h(y_{1},\ldots,y_{n}) |
\\
&& \quad \leq
\frac{1}{1+L} \sum_{i=p (\theta )+1}^{n} \bigl |
 \bigl\|x_{i} - f_{\theta}\bigl(x_{i-1},\ldots,x_{i-p (\theta )}\bigr)
 \bigr\|
-  \bigl\|y_{i} - f_{\theta}\bigl(y_{i-1},\ldots,y_{i-p (\theta
)}\bigr)  \bigr\| \bigr |
\\
&& \quad \leq
\frac{1}{1+L} \sum_{i=p (\theta )+1}^{n}
 \bigl\|x_{i}-y_{i} - f_{\theta}\bigl(x_{i-1},\ldots,x_{i-p (\theta
 )}\bigr)
+ f_{\theta}\bigl(y_{i-1},\ldots,y_{i-p (\theta )}\bigr) \bigr\|
\\
&& \quad \leq
\frac{1}{1+L} \sum_{i=p (\theta )+1}^{n}  \|
x_{i}-y_{i} \|
\\
&& \qquad {}+ \frac{1}{1+L} \sum_{i=p (\theta )+1}^{n}  \bigl\|
f_{\theta}\bigl(x_{i-1},\ldots,x_{i-p (\theta )}\bigr)
- f_{\theta}\bigl(y_{i-1},\ldots,y_{i-p (\theta )}\bigr) \bigr\|
\\
&& \quad \leq
\frac{1}{1+L} \sum_{i=p (\theta )+1}^{n}  \|
x_{i}-y_{i} \|
+ \frac{1}{1+L} \sum_{i=p (\theta )+1}^{n} \sum
_{j=1}^{p(\theta)} a_{j}(\theta)  \|x_{i-j}-y_{i-j} \|
\\
&& \quad \leq
\frac{1}{1+L} \sum_{i=p (\theta )+1}^{n}  \|
x_{i}-y_{i} \|
+ \frac{L}{1+L} \sum_{i=1}^{n}  \|x_{i}-y_{i} \|\\
&& \quad \leq\sum_{i=1}^{n}  \|x_{i}-y_{i} \|.
\end{eqnarray*}
The direct application of Theorem \ref{RIOthm} ends the proof under
{(WDP)}. Under {(CBS)}, $k_n$ follows from the estimates of $\|
X_0\|_\infty$ and $\theta_{\infty,n}(1)$ obtained in Proposition
\ref{lemma_ex_1}.
\end{pf*}

\begin{pf*}{Proof of Lemma \ref{lemrisk2}}
Integrate the inequality in Lemma \ref{lemrisk} with respect $\pi
_{p,\ell}$ on $\Theta_{p,\ell}$ (then $p(\theta)=p$) for any
$(p,\ell)\in M$ in order to obtain:
\[
\pi_{p,\ell}\bigl[ \pi_0\bigl[ \exp\bigl(
\lambda(\overline R- \overline r_n )\bigr)\bigr]\bigr] \leq
\exp
 \biggl(\frac{\lambda^{2} k_n^{2}}{n(1-p/n)^{2}} \biggr).
\]
Fubini's theorem implies that
\[
\pi_0 \biggl[ \pi_{p,\ell} \biggl[ \exp\biggl (
\lambda (\overline R- \overline r_n  )
-\frac{\lambda^{2} k_n^{2}}{n(1-p/n)^{2}} \biggr) \biggr] \biggr]
\leq1.
\]
Applying Lemma \ref{LEGENDRE} for $\pi=\pi_{p,\ell}$ and $h=\lambda
(\overline R- \overline r_n)-\lambda^{2} k_n^{2}/(n(1-p/n)^{2})$ on
$\mathcal{M}_{+}^{1}(\Theta_{p,\ell})$ leads to the inequality:
\[
\pi_0 \biggl[ \exp \biggl(\sup_{\rho\in\mathcal{M}_{+}^{1}(\Theta
_{p,\ell})}
 \{ \lambda\rho[\overline R- \overline r_n] - \mathcal{K}(\rho
,\pi_{p,\ell})  \} -\frac{\lambda^{2}
k_n^{2}}{n(1-p/n)^{2}} \biggr) \biggr]\leq1.
\]
This ends the proof.
\end{pf*}

\begin{pf*}{Proof of Lemma \ref{lemrisk3}}
First, let us choose $\lambda\in\Lambda$.
Let $h_{p,\ell}^{\lambda}$ denotes, for any $(p,\ell)\in M$:
\[
h_{p,\ell}^{\lambda}=\sup_{\rho_{p,\ell}\in\mathcal
{M}_{+}^{1}(\Theta_{p,\ell})}  \{ \lambda
\rho_{p,\ell}[\overline R- \overline r_n]
- \mathcal{K} (\rho_{p,\ell},\pi_{p,\ell} )  \}
-\frac{\lambda^{2} k_n^{2}}{n(1-p/n)^{2}}.
\]
From Lemma \ref{lemrisk2} applied on the different $\mathcal
{M}_{+}^{1}(\Theta_{p,\ell})$ we have,
for any $(p,\ell)\in M$:
\[
\pi_0 \biggl[\sum_{(p,\ell)\in M}w_{p,\ell} \exp (h_{p,\ell
}^{\lambda} ) \biggr]\leq1.
\]
Now we apply Inequality \eqref{lemmacatoni} in
Lemma \ref{LEGENDRE} for $\pi=\sum_{(p,\ell)\in M}w_{p,\ell}\delta
_{(p,\ell)}$
and $h=\break \sum_{(p,\ell)\in M}h_{p,\ell}^{\lambda}\1_{\Theta_{p,\ell
}}$ and we obtain
\[
\pi_0 \biggl[\exp \biggl(\sup_{\sum_{(p,\ell)\in M}w'_{p,\ell
}=1} \biggl\{
\sum_{(p,\ell)\in M}w'_{p,\ell} h_{p\ell}-\sum_{(p,\ell)\in
M}w'_{p,\ell}\log(w'_{p,\ell}/w_{p,\ell}) \biggr\} \biggr)
\biggr]\le1
\]
and, by Jensen's inequality, and replacing $h_{p,\ell}^{\lambda}$ by
its definition,
%
\begin{eqnarray}\label{eq::a}
&&\hspace*{-8pt}\pi_0\biggl [\sup_{\sum_{(p,\ell)\in M}w'_{p,\ell}=1} \biggl\{
\sum_{(p,\ell)\in M}w'_{p,\ell} \sup_{\rho_{p,\ell}\in\mathcal
{M}_{+}^{1}(\Theta_{p,\ell})} \exp \biggl(
\lambda\rho_{p,\ell}  \biggl[\lambda(\overline R- \overline r_n)
- \log\frac{\mathrm{d}\rho_{p,\ell}}{\mathrm{d}\pi_{p,\ell}} \biggr]
\nonumber
\\[-8pt]
\\[-8pt]
&&\hspace*{-8pt}\hspace*{207pt}{} -\frac{\lambda^{2} k_n^{2}}{n(1-p/n)^{2}}
+\log\frac{w_{p,\ell}}{ w'_{p,\ell}} \biggr) \biggr\} \biggr] \le1.
\nonumber
\end{eqnarray}
By Jensen again, we obtain a bound for the first term in the sum
bounded in Lemma \ref{lemrisk3}:
\begin{eqnarray*}
&&\hspace*{-4pt}\pi_0 \biggl[\sup_{\sum_{(p,\ell)\in M}w'_{p,\ell}=1} \biggl\{
\sum_{(p,\ell)\in M}w'_{p,\ell}
\sup_{\rho_{p,\ell}\in\mathcal{M}_{+}^{1}(\Theta_{p,\ell})} \rho
_{p,\ell} \biggl[\exp \biggl(
\lambda(\overline R- \overline r_n)
- \log\frac{\mathrm{d}\rho_{p,\ell}}{\mathrm{d}\pi_{p,\ell}}     \\
&&\hspace*{-6pt}\hspace*{229pt}{}-\frac{\lambda^{2} k_n^{2}}{n(1-p/n)^{2}}
+\log\frac{w_{p,\ell}}{ w'_{p,\ell}} \biggr) \biggr] \biggr\}
\biggr] \le1.
\end{eqnarray*}
Finally, we sum this inequality over all $\lambda\in\mathcal G$ to
bound the first expectation.

The second expectation is bounded by choosing specific weights
$w'_{p,\ell}$ in the supremum in inequality \eqref{eq::a} such that
$w'_{p,\ell}=1$ for $(p,\ell)= \arg\max_M\{h_{p,\ell}\}$:
\[
\pi_0 \biggl[
\mathop{\mathop{\sup}_{
(p,\ell)\in M}}_{
\rho_{p,\ell}\in\mathcal{M}_{+}^{1}(\Theta_{p,\ell})
}
\biggl \{
\exp \biggl(
\lambda \rho_{p,\ell} [\overline R- \overline r_n]
- \mathcal{K}(\rho_{p,\ell},\pi_{p,\ell})-\frac{\lambda^{2}
k_n^{2}}{n(1-p/n)^{2}}
+\log w_{p,\ell} \biggr) \biggr\} \biggr] \le1.
\]
Again a summation over all $\lambda\in\mathcal G$ leads to the
result. This ends the proof.
\end{pf*}

\begin{pf*}{Proof of Lemma \ref{lembound}}
From the proof of the Lemma \ref{lemrisk}, we already know that
$|\overline r_n(\theta)-r_n(\theta)|\le (1+L)/(n-p)\sum_{i=1}^n\|
X_i-\overline X_i\|$. This bound holds uniformly on $\Theta$. As $p\le
n/2$ it remains to estimate $\pi_0[\exp(\lambda2 (1+L)/n\sum
_{i=1}^n\|X_i-\overline X_i\|])$. From the assumption \eqref{condlip},
the stationarity of $X$ and as the $\xi_i$s are i.i.d. we have:
\begin{eqnarray*}
&&\pi_0 \Biggl[\exp \Biggl( \lambda2(1+L)/n\sum_{i=1}^n\|X_i-\overline
X_i\| \Biggr)\Biggr] \\
&& \quad \le\pi_0 \Biggl[\exp \Biggl(\lambda2(1+L)/n\sum
_{i=1}^n\sum_{j=0}^\infty a_j(H)\|\xi_{i-j}-\overline\xi_{i-j}\|
 \Biggr) \Biggr]\\
&& \quad \le\pi_0 \Biggl[\exp \Biggl( \lambda2 (1+L)/n\sum_{j=0}^\infty\sum
_{i=1\vee(n-j)}^na_{n-i+j}(H)\|\xi_{n-j}-\overline\xi_{n-j}\|
\Biggr) \Biggr]\\
&& \quad \le\prod_{j=0}^\infty\pi_0 \Biggl[\exp\Biggl( \lambda2(1+L)/n\sum
_{i=1\vee(n-j)}^na_{n-i+j}(H)\|\xi_0\|\1_{\|\xi_0\|>C}\Biggr) \Biggr] .
\end{eqnarray*}
Denoting $c_j= \lambda2(1+L)\sum_{i=1\vee(n-j)}^na_{n-i+j}(H)/n$, we
develop for all $j\ge0$
\[
\pi_0\bigl [\exp\bigl( c_j\|\xi_{0}\|\1_{\|\xi_0\|>C}\bigr) \bigr]=1+c_j\pi
_0 \bigl[\|\xi_{0}\|\1_{\|\xi_0\|>C} \bigr]+\sum_{k\ge2}\frac
{c_j^k\pi_0 [\|\xi_{0}\|^k\1_{\|\xi_0\|>C} ]}{k!}.
\]
As $ \Psi(a(H))=\pi_0[\exp(a(H)\|\xi_0\|)]=\sum_{k\ge0}a(H)^k\pi
_0[\|\xi_{0}\|^k ]/k!$ is a convergent series of sequence of positive
numbers, one gets
\[
\pi_0\bigl [\|\xi_{0}\|^k\1_{\|\xi_0\|>C} \bigr]\le\pi_0 [\|\xi
_{0}\|^k  ]\le\frac{k! \Psi(a(H))}{{a(H)}^k} \qquad  \forall k\ge2.
\]
As $\lambda< n/(4(1+L))$ then $2c_j\le a(H)$ for all $j\ge0$ and then
we derive that for all $j\ge0$:
\begin{eqnarray*}
\pi_0 \bigl[\exp\bigl( c_j\|\xi_{0}\|\1_{\|\xi_0\|>C}\bigr) \bigr]&\le&1+c_j\pi
_0 \bigl[\|\xi_{0}\|\1_{\|\xi_0\|>C} \bigr]+\Psi(a(H))\sum_{k\ge
2}(c_j/a(H))^k\\
&\le&1+c_j\pi_0 \bigl[\|\xi_{0}\|\1_{\|\xi_0\|>C} \bigr]+\frac{\Psi
(a(H))c_j^2}{a(H)(a(H)-c_j)}\\
&\le&1+c_j\pi_0 \bigl[\|\xi_{0}\|\1_{\|\xi_0\|>C} \bigr]+c_j^2\frac
{2\Psi(a(H))}{a(H)^2}.
\end{eqnarray*}
As $\phi(x)= (\exp(x)-1)/x$ is an increasing function for $x>0$, then
$\1_{\|\xi_0\|>C}\le\phi(a(H)\|\xi_0\|)/\break \phi(a(H)C)$ and the
Markov formula gives for all $j\ge0$
\[
\pi_0 \bigl[\exp\bigl( c_j\|\xi_{0}\|\1_{\|\xi_0\|>C}\bigr) \bigr]\le1+ c_j
\frac{\Psi(a(H))a(H)C}{\exp(a(H)C)-1}+c_j^2\frac{2\Psi(a(H))}{a(H)^2}.
\]
Collecting those bounds, we obtain
\[
\pi_0 \Bigl[\exp \Bigl( \lambda\sup_{\theta\in\Theta}|\overline
r_n(\theta)-r_n(\theta)| \Bigr) \Bigr]\le\prod_{j=0}^\infty
\biggl(1+c_j \frac{\Psi(a(H))a(H)C}{\exp(a(H)C)-1}+c_j^2\frac{2\Psi
(a(H))}{a(H)^2} \biggr).
\]
Using that $\log(1+x)\le x$ for all $x>0$, we finally obtain:
\[
\log\Bigl (\pi_0 \Bigl[\exp \Bigl( \lambda\sup_{\theta\in\Theta
}|\overline r_n(\theta)-r_n(\theta)| \Bigr) \Bigr] \Bigr)\le\sum
_{j=0}^\infty c_j \frac{\Psi(a(H))a(H)C}{\exp(a(H)C)-1}+\sum
_{j=0}^\infty c_j^2\frac{2\Psi(a(H))}{a(H)^2}.
\]
The desired result follows from the estimates $\sum_{j=0}^\infty
c_j\le\lambda a(H)2(1+L)$ and $\sum_{j=0}^\infty c_j^2\le\lambda^2
a(H)^24(1+L)^2/n$.
\end{pf*}

Now give the proof of the useful Proposition \ref{propdim}.
\begin{pf*}{Proof of Proposition \ref{propdim}}
Let us introduce a parameter $\zeta>0$ then we have
\begin{eqnarray*}
-\frac{1}{\gamma}\log\pi_{p,\ell} \bigl[\exp\bigl ( -\gamma
\bigl(R - R(\overline{\theta}_{p,\ell}) \bigr) \bigr) \bigr]-\zeta
&=&
-\frac{1}{\gamma}\log\pi_{p,\ell} \bigl [\exp \bigl(-\gamma
\bigl(R -R(\overline{\theta}_{p,\ell})-\zeta \bigr) \bigr) \bigr]\\
&\leq&
-\frac{1}{\gamma}\log\pi_{p,\ell}  \bigl( R(\theta) - R(\overline
{\theta}_{p,\ell}) \leq\zeta \bigr).
\end{eqnarray*}
Then we directly derive from the definition of $d_{p,\ell}$ that
\[
d_{p,\ell}\le\sup_{\gamma>e} \frac{\inf_{\zeta>0}\{\zeta\gamma
-\log\pi_{p,\ell}  ( R(\theta) -
R(\overline{\theta}_{p,\ell}) \leq\zeta )\}}{\log\gamma}.
\]
So
\[
\zeta\gamma-q\log\frac{\zeta}{C c_{p,\ell}}\le q\wedge\gamma
C(c_{p,\ell}-\|\overline\theta_{p,\ell}\|)+ q \log \biggl(\frac
{Cc_{p,\ell}\gamma}{q}\vee\frac{c_{p,\ell}}{c_{p,\ell}-\|
\overline\theta_{p,\ell}\|} \biggr).
\]
Now if $q\le\gamma C(c_{p,\ell}-\|\overline\theta_{p,\ell}\|)$
then we get the estimate
$q(1+\log(Cc_{p,\ell}\gamma/ q))/\log\gamma$
which decreases with $\gamma$. We then get the desired bound when the
supremum is established for $\gamma=e\vee q/(C(c_{p,\ell}-\|\overline
\theta_{p,\ell}\|))$. If $q\ge\gamma C(c_{p,\ell}-\|\overline
\theta_{p,\ell}\|)$, then we get the estimate
$
(\gamma C(c_{p,\ell}-\|\overline\theta_{p,\ell}\|)+q\log
(c_{p,\ell}/(c_{p,\ell}-\|\overline\theta_{p,\ell}\|)))/\log
\gamma
$ which increases with $\gamma$. We have to consider $\gamma$ as
large as possible, that is, when $q = \gamma C(c_{p,\ell}-\|\overline
\theta_{p,\ell}\|)$ and we are going back to the case treated above.\looseness=1
\end{pf*}

\subsection{Proofs of the results given in Section \protect\ref{sec::app1}}

\label{proofobservations}

After proving Proposition \ref{lemma_ex_1}, we give Lemma \ref
{couplinglemma} that introduces a coupling argument used to estimate
the coefficients $\theta_{\infty,n}(1)$ in Propositions \ref{taucbs}
and \ref{lemma_ex_2}.

\begin{pf*}{Proof of Proposition \ref{lemma_ex_1}}
The Theorem 3.1 of Doukhan and Wintenberger \cite{Doukhan2008} gives
the existence of a unique stationary solution and the existence of an
$H$ such that $X_t=H(\xi_t,\xi_{t-1},\xi_{t-2},\break \ldots)$. We prove
that conditions \eqref{condlip1} and \eqref{condsum} are
automatically satisfied. Let $(x_i)$ and $(y_i)$ be two sequences such
that there exists $j\in\N$ with $x_i=y_i$ for all $i\neq j$. Then
$H(x)=u_0^\infty$ where $u_0^\infty=\lim_{k\to\infty}u_0^k$ for
$(u_{-i}^k)_{i\in\N}$ defined recursively by
\[
u_{-i}^k=F(u_{-i-1}^k,u_{-i-2}^k,\ldots, u_{1-k}^k,u_{-k}^k,0,\ldots
;x_i) \qquad  \forall0\le i\le k.
\]
Similarly, we denote $H(y)=v_0^\infty$ such that $\|H(x)-H(y)\|=\|
u_0^\infty-v_0^\infty\|$. For $j=0$, using \eqref{condlip} $\|
u_{0}^k-v_{0}^k\|\le u\|x_j-y_j\|$ for all $k$. For $j\ge1$, as
$x_i=y_i$ for $i>j$, for $k$ sufficiently large it holds (with the
convention $\sum_{\ell=1}^{-k}=0$ for $k\ge0$):
\[
\|u_0^k-v_0^k\|\le \sum_{\ell_1 = 1}^j a_{\ell_1}(F)\sum_{\ell_2
= 1}^{j-\ell_1}a_{\ell_2}(F)\cdots\sum_{\ell_j = 1}^{j-\ell
_1-\cdots-\ell_{j-1}}a_{\ell_j}(F)\|u_{-j}^k-v_{-j}^k\|.
\]
By definition $\|u_{-j}^k-v_{-j}^k\|\le u\|x_j-y_j\|$ and we obtain
$\|u_0^k-v_0^k\|\le ua(F)^{j-1}\|x_j-y_j\|$ for sufficiently large $k$.
As the estimate does not depends on $k$, we derive that \eqref
{condlip1} holds with $a_j(H)=ua(F)^{j-1}$ and that \eqref{condsum}
follows from the condition \eqref{condcontract}.
\end{pf*}

Now we state a useful coupling lemma; $(X^\ast_t) $ is said to be a
coupling version of $(X_t)$ if it is similarly distributed and such
that $(X^\ast_t)_{t>0}$ is independent of $\mathfrak{S}_0=\sigma
(X_t,t\le0)$. From a version of the Kantorovitch--Rubinstein duality,
see Dedecker and Prieur \cite{Dedecker2005} for more details, we
obtain an estimate of $\theta_{\infty,n}(1)$.
\begin{lemma}\label{couplinglemma}
For any version $(X^\ast_t)$, we have
\[
\theta_{\infty,n}(1)\le\sum_{i=1}^n\bigl\|\E(\|X_i-X_i^\ast\|
/\mathfrak{S}_0)\bigr\|_\infty.
\]
\end{lemma}

For the sake of completeness, we give the proof of this already known
result.

\begin{pf*}{Proof of Lemma \ref{couplinglemma}} As we equipped
$\mathcal{X}^n$ with the norm $\|(x_1,\ldots,x_n)\|=\sum_{i=1}^n \|
x_i\|$, we immediately get the inequality
\[
{\theta_{\infty,n}}(1)\le\bigl\|\mathds{\mathds{E}}\bigl(\|(X_1,\ldots
,X_n)-(X_1^\ast,\ldots,X_n^\ast)\||\mathfrak{S}_0\bigr)\bigr\|_\infty
\le\sum_{i=1}^u\|\mathds{\mathds{E}}(\|X_i-X_i^\ast\||\mathfrak
{S}_0)\|_\infty.
\]
\upqed
\end{pf*}

The proof of Propositions \ref{taucbs} and \ref{lemma_ex_2} are
simple applications of this lemma.
\begin{pf*}{Proof of Proposition \ref{taucbs}}
Let us consider the coupling version of the causal Bernoulli shift
$(X_t)$ given by
\[
X_t^\ast=H(\xi_t,\xi_{t-1},\ldots,\xi_{1},\xi_0^\ast,\xi
_{-1}^\ast,\ldots) \qquad  \forall t\in\Z,
\]
where $(\xi_t^\ast)$ is similarly distributed than $(\xi_t)$ and the
two processes are independent. Then from Lemma \ref{couplinglemma} and
condition \eqref{condlip1}, we obtain:
\[
\theta_{\infty,n}(1)\le\sum_{i=1}^n \Biggl\|\sum_{j=i}^\infty
a_j(H)\E(\|\xi_{i-j}-\xi_{i-j}^\ast\|/\mathfrak{S}_0) \Biggr\|
_\infty\le\sum_{j=i}^\infty j a_j(H)\bigl\|\E(\|\xi_{i-j}-\xi
_{i-j}^\ast\|/\mathfrak{S}_0) \bigr\|_\infty
\]
and the desired result follows.
\end{pf*}

\begin{pf*}{Proof of Proposition \ref{lemma_ex_2}}
Here we will consider the maximal coupling scheme of Goldstein \cite
{Goldstein1979}: there exists a version $(X_t^\ast)$ such that
\[
\|\P(X_t\neq X_t^\ast\mbox{ for some }t\ge r/ \mathfrak{S}_0)\|
_\infty=\sup_{(A,B)\in \mathfrak{S}_0\times\mathfrak{F}_r}|\P
(A/B)-P(B)|=\varphi(r).
\]
As $\|Y-Z\|\le2\|X_0\|_\infty\1_{Y\neq Z}$ for any variables $Y,Z$
bounded by $\|X_0\|_\infty$, we have:
\[
\bigl\|\mathds{\mathds{E}}(\|X_i-X_i^\ast\|/ \mathfrak{S}_0)\bigr\|_\infty
\le2\|X_0\|_\infty\|\mathds{\mathds{E}}(\1_{X_i\neq X_i^\ast}/
\mathfrak{S}_0)\|_\infty
\le2\|X_0\|_\infty\|\mathds{\mathds{P}}(X_i\neq X_i^\ast/
\mathfrak{S}_0)\|_\infty.
\]
As
$
\mathds{P}(X_i\neq X_i^\ast/ \mathfrak{S}_0)\le\mathds{P}( X_t\neq
X_t^\ast\mbox{ for some }t\ge r/ \mathfrak{S}_0)
$,
we conclude using Lemma \ref{couplinglemma}.
\end{pf*}

\subsection{Proofs of the results given in Section \protect\ref{sec::app2}}

\label{proof::app}
We proof the Corollaries \ref{cor::nn} and \ref{fourier} of Theorem
\ref{mainthm} applied in the context of Neural Networks and projection
in the Fourier basis predictors.

\begin{pf*}{Proof of Corollary \ref{cor::nn}}\label{proofnn}
Let us check that all the predictors are $L$-Lipschitz functions of the
observations. For any $x,y\in\R^p$, as the function $\phi$ is
1-Lipschitz, we have
\begin{eqnarray*}
|f_\theta(x)-f_\theta(y)|&\le& \Biggl|\sum_{k=1}^\ell c_k\bigl(\phi
(a_k\cdot x+b_k)-\phi(a_k\cdot y+b_k)\bigr) \Biggr|\\
&\le& \sum_{k=1}^\ell|c_k||a_k\cdot(x-y)|
\le \sum_{k=1}^\ell|c_k|\|a_k\|_1\|x-y\|_\infty\\
&\le& \|\|a_k\|_1\|_\infty\sum_{k=1}^\ell|c_k|\sum_{i=1}^p|x_i-y_i|.
\end{eqnarray*}
For $\theta\in\mathcal B^q_{c_{p,\ell}}$ then $L=(c_{p,\ell}\vee
1)^3 $ is convenient. Next, using Jensen to estimate $\mathbb L_1$-risk
by $\mathbb L_2$-risk, we obtain from the Theorem 1 of Barron \cite
{Barron1994} the existence of $C>0$ such that
\[
\pi_0 [ |\operatorname{med}(X_0 | X_{-1},\ldots
,X_{-p})-f_{\overline\theta_{p,\ell}}(X_{-1},\ldots,X_{-p})
| ]\le C \frac{ p^c\|X_0\|_\infty}{\sqrt\ell},
\]
where $\overline\theta_{p,\ell}$ belongs to the compact set
\[
\mathcal B'_{p,\ell}=\Biggl \{
\theta\in\R^{\ell(p+2)+1} ; \sum_{i=1}^\ell|c_i|\le C'c^p;
\max_{1\le i\le\ell}\|a_i\|\le\sqrt\ell\log\ell;
  \max_{1\le i\le\ell}|b_i|\le\|X_0\|_\infty\sqrt\ell\log
\ell \Biggr\}.
\]
Remark that under the assumptions of Corollary \ref{cor::nn}, we have
$c_{p,\ell}- \|\overline\theta_{p,\ell}\|\ge q/e$. It implies by
Proposition \ref{propdim} that
$d_{p,\ell}\le3q(1+\log(c_{p,\ell}))$ when $c_{p,\ell}\ge1$. From
Theorem \ref{mainthm} there exists $C>0$ satisfying
\[
R(\hat{\theta})]\le\inf_{d_{p,\ell}\le n} \Biggl \{R^\ast_p
+C \Biggl( \frac{p^c}{\sqrt\ell} + \log^{3} (n) \sqrt\frac{p\ell
}{n} \Biggr) \Biggr\}+C\frac{\log\fracd{1}{\varepsilon}}{\sqrt{n}}.
\]
The result follows from considering $\ell=\sqrt n p^{c-1/2}$.
\end{pf*}

\begin{pf}{Proof of Proposition \ref{fourier}}\label{prooffourier}
Let us apply Theorem \ref{mainthm2}: there exists $C>0$ such that
\begin{eqnarray*}
R(\hat{\theta})&\le&
\inf_{  p,\ell: d_{p,\ell}\leq n}\Biggl  \{
\min_{\theta\in\Theta_{p,\ell}}R (\theta ) +
C \sqrt{\frac{d_{p,\ell}}{n}} \log^{5/2}(n)  \Biggr\}
+ C
\frac{\log\fracd{1}{\varepsilon}}{\sqrt{n}}
\\
&\leq&
\inf_{ \ell: d_{p_{0},\ell}\leq n}  \Biggl\{
\min_{\theta\in\Theta_{p_{0},\ell}}R (\theta ) +
C \sqrt{\frac{d_{p_{0},\ell}}{n}} \log^{5/2}(n)  \Biggr\}
+ C
\frac{\log\fracd{1}{\varepsilon}}{\sqrt{n}}
.
\end{eqnarray*}
Remarking that
\begin{eqnarray*}
R (\overline{\theta}_{p_{0},\ell} )
&=& \inf_{\theta\in\Theta}{\pi_{0}}
 [  |X_{p+1}-f_{\overline{\theta}_{p_{0},\ell
}}(X_{p},\ldots,X_{1}) | ]
\\
&\leq& {\pi_{0}}  \Biggl[  \Biggl|X_{p+1}-\sum_{i=1}^{p_{0}}
f_{i}(X_{p-i})  \Biggr| \Biggr] + \inf_{\theta\in\Theta}
{\pi_{0}} \Biggl [ \Biggl| \sum_{i=1}^{p_{0}}
f_{i}(X_{p-i}) - \sum_{i=1}^{p_{0}} \sum_{j=1}^{n} \theta_{i,j}
\varphi_{j}(X_{p-i})  \Biggr| \Biggr]
\\
&\leq&\mu[|\xi_{0}|] + \inf_{\theta\in\Theta} \sum_{i=1}^{p_{0}}
{\pi_{0}}  \Biggl[  \Biggl| f_{i}(X_{1}) - \sum_{j=1}^{n}
\theta_{i,j} \varphi_{j}(X_{1})  \Biggr| \Biggr] .
\end{eqnarray*}
Note also that under our hypothesis $X_{1}$
has a density upper bounded by $1/\sqrt{2\uppi\sigma^{2}}$. It then holds
\begin{eqnarray*}
R (\overline{\theta}_{p_{0},\ell} ) &\leq&\mu[|\xi_{0}|] +
\frac{1}{\sqrt{2\uppi\sigma^{2}}} \inf_{\theta\in\Theta}
\sum_{i=1}^{p_{0}} \int  \Biggl| f_{i}(x) - \sum_{j=1}^{n}
\theta_{i,j} \varphi_{j}(x)  \Biggr|\,\mathrm{d}x
\\
&\leq&\mu[|\xi_{0}|] + \frac{1}{\sqrt{2\uppi\sigma^{2}}}
\inf_{\theta\in\Theta} \sum_{i=1}^{p_{0}}  \Biggl(\int  \Biggl[
f_{i}(x) - \sum_{j=1}^{n} \theta_{i,j} \varphi_{j}(x)  \Biggr]^{2}\,\mathrm{d}x \Biggr)^{\fraca{1}{2}}
\\
&\leq&
\mu[|\xi_{0}|] + \frac{1}{\sqrt{2\uppi\sigma^{2}}} \sum_{i=1}^{p_{0}}
\gamma_{i} \ell^{-s_{i}}
\leq\mu[|\xi_{0}|] + \frac{\sum_{i=1}^{p_{0}}\gamma_{i}}{\sqrt
{2\uppi\sigma^{2}}}
\ell^{-s}.
\end{eqnarray*}
Then we have
%
\begin{equation}
\label{kaka}
\pi_{0}[R(\hat{\theta})]\le\mu[|\xi_{0}|] + \inf_{\ell} \Biggl \{
\ell^{-s} \frac{\sum_{i=1}^{p_{0}}\gamma_{i}}{\sqrt{2\pi\sigma^{2}}}
+ C \sqrt{\frac{d_{p_{0},\ell}}{n}} \log^{5/2}(n)
 \Biggr\}
+C
\frac{\log\fracd{1}{\varepsilon}}{\sqrt{n}}.
\end{equation}
The estimate of $d_{p_{0},\ell}$ from Proposition \ref{propdim} is
plugged into
\eqref{kaka} to obtain for some $C>0$
\[
\pi_{0}[R(\hat{\theta})]\le\mu[|\xi_{0}|] + \inf_{\ell} \Biggl \{
\ell^{-s} \frac{\sum_{i=1}^{p_{0}}\gamma_{i}}{\sqrt{2\uppi\sigma^{2}}}
+ C \sqrt{\frac{p_{0}\ell}{n}} \log^{5/2}(n)
 \Biggr\}
+ C
\frac{\log\fracd{1}{\varepsilon}}{\sqrt{n}}.
\]
In particular, fixing $\ell$ proportional to $n^{\fracc{1}{2s+1}}$
leads to the result.
\end{pf}


\section*{Acknowledgements}
We would like to thank the anonymous referees for the various
corrections and improvements
they suggested.

%

\printhistory

\end{document}